\documentclass[a4paper]{article}

\usepackage{a4wide}
\usepackage{graphicx}
\graphicspath{{}}
\usepackage{amsmath}
\usepackage{amsthm}
\usepackage{amssymb, latexsym, mathrsfs}
\usepackage{enumerate}
\usepackage{algorithm}
\usepackage{subcaption}
\usepackage{color}
\usepackage{transparent}
\usepackage{url}
\usepackage[affil-it]{authblk}
\theoremstyle{plain}
\newtheorem{thm}{Theorem}
\newtheorem{assum}{Assumption}

\newtheorem{rem}{Remark}
\newtheorem{prop}{Proposition}
\newtheorem{defi}{Definition}

\newcommand{\Gset}{\mathbb{G}}
\newcommand{\Hset}{\mathbb{H}}

\newcommand{\Oset}{\mathbb{O}}
\newcommand{\Pset}{\mathbb{P}}

\newcommand{\Rset}{\mathbb{R}}
\newcommand{\Sset}{\mathbb{S}}

\newcommand{\Uset}{\mathbb{U}}
\newcommand{\Vset}{\mathbb{V}}
\newcommand{\Wset}{\mathbb{W}}
\newcommand{\Xset}{\mathbb{X}}

\newcommand{\Zset}{\mathbb{Z}}

\newcommand{\hXset}{\mathbb{\hat{X}}}

\newcommand{\bZset}{\mathbb{\bar{Z}}}

\newcommand{\tXset}{\mathbb{\tilde{X}}}

\newcommand{\hx}{{\hat{x}}}

\newcommand{\bv}{{\bar{v}}}

\newcommand{\bx}{{\bar{x}}}

\newcommand{\tg}{{\tilde{g}}}
\newcommand{\tildeh}{{\tilde{h}}}

\newcommand{\tw}{{\tilde{w}}}
\newcommand{\tx}{{\tilde{x}}}

\newcommand{\CC}{{\mathcal{C}}}

\newcommand{\FF}{{\mathcal{F}}}

\newcommand{\HH}{{\mathcal{H}}}

\newcommand{\MM}{{\mathcal{M}}}
\newcommand{\NN}{{\mathcal{N}}}

\newcommand{\PP}{{\mathcal{P}}}

\newcommand{\dist}[2]{\mbox{dist}({#1},{#2})}      
\newcommand{\abs}[1]{{|{#1}|}}                     
\newcommand{\norme}[2]{||{#1}||_{#2}}            
\newcommand{\ball}[1]{{B_{#1}}}                    
\newcommand{\convh}{\mbox{convh}}                  

\newcommand{\imply}{\Rightarrow}                   

\newcommand{\mbf}[1]{\mathbf{#1}}                  

\newcommand{\Zero}{\textbf{0}}

\newcommand{\subss}[2]{ #1_{[#2]} }              

\newcommand{\xp}{x^+}                              

\newcommand{\txp}{{\tilde{x}}^+}                   

\newcommand{\hxp}{{\hat{x}}^+}                     

\newcommand{\bkappa}{{\bar\kappa}}                 

\newcommand{\matr}[1]{
\begin{bmatrix}
    #1
\end{bmatrix}
}

\begin{document}
      \title{\LARGE \bf Plug-and-play fault diagnosis and control-reconfiguration for a class of nonlinear large-scale constrained systems\thanks{The research leading to these results has received funding from the European Union Seventh Framework Programme [FP7/2007-2013]  under grant agreement n$^\circ$ 257462 HYCON2 Network of excellence.}\thanks{Electronic address: \texttt{stefano.riverso@unipv.it}; Corresponding author}}

    \author[1]{Stefano Riverso%
       } 

     \author[2]{Francesca Boem%
       }

     \author[1]{Giancarlo Ferrari-Trecate%
       }
        
     \author[2,3]{Thomas Parisini%
       } 

     \affil[1]{Universit\`a degli Studi di Pavia, Italy}
     \affil[2]{Universit\`a degli Studi di Trieste, Italy}
     \affil[3]{Imperial College London, UK}
     
     \date{\textbf{Technical Report}\\ September, 2014}

     \maketitle

     \begin{abstract}
       This paper deals with a novel Plug-and-Play (PnP) architecture for the control and monitoring of Large-Scale Systems (LSSs). The proposed approach integrates a distributed Model Predictive Control (MPC) strategy with a distributed Fault Detection (FD) architecture and methodology in a PnP framework. The basic concept is to use the FD scheme as an autonomous decision support system: once a fault is detected, the faulty subsystem can be unplugged to avoid the propagation of the fault in the interconnected LSS. Analogously, once the issue has been solved, the disconnected subsystem can be re-plugged-in. PnP design of local controllers and detectors allow these operations to be performed safely, i.e. without spoiling stability and constraint satisfaction for the whole LSS. The PnP distributed MPC is derived for a class of nonlinear LSS and an integrated PnP distributed FD architecture is proposed. Simulation results show the effectiveness and the potential of the general methodology.
     \end{abstract}

     \newpage

     \section{Introduction}

          Nowadays, several man-made systems are characterized by a large number of states and inputs with a significant spatial distribution. This triggered an increasing interest in the study of Systems-of-Systems \cite{Samad2011} and Cyber-Physical Systems \cite{Baheti2011}. LSSs are often modeled as the interaction of many subsystems coupled through physical variables or communication channels \cite{Lunze1992}. When dealing with control of LSSs, centralized control architectures can be impractical due to computational, communication and reliability limits, and an alternative is offered by the adoption of decentralized and distributed approaches.

          In the past, several decentralized (De) and distributed (Di) MPC schemes have been proposed for constrained LSS (see the recent survey \cite{Christofides2013} and references therein). In the standard MPC control of LSSs, the prediction of the LSS behaviour is carried out through a nominal model of each subsystem and of the local interactions. However, in several applications, faults and malfunctions may occur thus possibly causing critical and unpredictable changes in the LSS dynamics. Hence, there is a need to devise fault diagnosis schemes (see, for example, \cite{Blanke2003,Isermann2006}) providing on-line the information about the health of the system and to exploit this information to reconfigure the controller so as to guarantee some degree of fault-tolerance (see \cite{Zhang2004}). Model-based schemes have emerged as prominent approaches to fault diagnosis of continuous and discrete-time systems \cite{Venkatasubramanian2003}. As for centralized control, centralized FD architectures suffer of scalability and robustness issues. In this context, decentralized and distributed fault-tolerant control and fault diagnosis algorithms have been proposed (see \cite{Patton2007}, \cite{Li2009}, \cite{Zhang2012}, \cite{Boem2011}, \cite{Boem2011b}, \cite{Ferrari2012} as examples).
 
     In this paper, the integration of a DiMPC scheme and a distributed FD architecture is proposed for the first time. Specifically, in the off-line control design phase we adopt a decentralized algorithm and we assume that the design of a local controller can use information at most from parents of the corresponding subsystem, i.e., subsystems that influence its dynamics. This implies that the whole model of the LSS is never used in any step of the synthesis process \cite{Lunze1992}. This approach has several advantages in terms of {\em scalability}: i) the communication flow at the design phase has the same topology of the coupling graph, usually sparse; ii) the local design of controllers and fault detectors can be conducted in parallel; iii) local design complexity scales with the number of parent subsystems only; iv) if a subsystem joins/leaves an existing network (plug-in/unplugging operation) at most children/parents subsystems have to retune their controllers and fault detectors. We refer to this kind of decentralized synthesis as PnP design, if, in addition, the plug-in and unplugging operations can be performed through a procedure for automatically assessing whether the operation does not spoil stability and constraint satisfaction for the overall LSS (see \cite{Riverso2013c} and \cite{Riverso2014a}). Different definitions of PnP design are given in \cite{Stoustrup2009}, \cite{Bendtsen2013} and \cite{Bodenburg2014}.
               
     \subsubsection*{Novelties}
          The very significant novelty presented in the paper is the {\em integration of DiMPC and FD architectures in a PnP framework}. Similarly to the design of local controllers, we propose a PnP design method for local fault detection. Motivations for PnP MPC/FD are the following: i) when the behaviour of a subsystem is corrupted by a fault, we show how the subsystem can be automatically disconnected while preserving stability and constraint satisfaction at each time instant for all other subsystems; ii) when a faulty subsystem is repaired, it can be replugged-in without changing all existing local controllers and fault detectors. We highlight that, differently from \cite{Riverso2013c} and \cite{Riverso2014a}, in this paper we design local MPC controllers for a class of nonlinear LSS. As regards FD schemes, to the best of the authors knowledge, it is the first time that a PnP FD architecture is proposed. Furthermore, in real application contexts, usually, MPC controllers are designed based on the knowledge of a nominal model of the system. Therefore a FD scheme is needed to monitor the behaviour of the system. The proposed FD architecture is robust to modeling and measurement uncertainties. To achieve this goal, it considers local models that are different from those used in local MPC controllers. In fact, another novel contribution of this paper is the possibility to use different decompositions and different modeling for the control and the diagnosis frameworks. This feature is useful for applications: local controllers must compute local control inputs based on local available measurements only, sometimes with high sampling rates; on the other hand local fault detectors may work at a different rate and can keep advantage of the redundancy given by sharing some variables in order to improve estimation performances. It is worth noting that, to the best of the authors knowledge, this is the first contribution addressing distributed schemes for nonlinear LSS integrating model-based fault diagnosis with MPC. For centralized approaches, the interested reader is referred to \cite{Prakash2005}, \cite{Raimondo2013} and the related work in \cite{Scott2014}. Moreover, a centralized reconfiguration process, based on hybrid systems, is proposed in \cite{Tsudal2001}.

          A preliminary version of this work, without any theoretical proofs and simulation examples, has been accepted at the 53rd IEEE Conference on Decision and Control \cite{Riverso2014b}.
          
          The paper is organized as follows. In Section~\ref{sec:systemdef}, we define the problem dealt with in the paper and we introduce the dual decomposition of the LSS. Then, in Section~\ref{sec:nonlinearsystube}, we design the nonlinear DiMPC architecture, while in Section \ref{sec:fault_architecture} we derive the PnP distributed FD scheme. The fault detectability analysis is presented in Section~\ref{sec:faulty}. The reconfiguration process after unplugging and plugging-in operations are described in Section \ref{sec:reconfiguration}. In Section~\ref{sec:simulationExample}, we apply the proposed architectures to a ring of coupled van der Pol Oscillators (vdPOs) and to a Power Network System (PNS). Finally, some concluding remarks are given in Section~\ref{sec:conclusion}. 

          \textbf{Notation.} We use $a:b$ for the set of integers $\{a,a+1,\ldots,b\}$. The column vector with $s$ components $v_1,\dots,v_s$ is $\mbf v=(v_1,\dots,v_s)$. The symbols $\oplus$ and $\ominus$ denote the Minkowski sum and difference, respectively, i.e. $A=B\oplus C$ if $A=\{a:a=b+c,\mbox{ for all }b\in B \mbox{ and }c\in C\}$ and $A=B\ominus C$ if $a\oplus C\subseteq B,~\forall a\in A$. Moreover, $\bigoplus_{i=1}^sG_i=G_1\oplus\ldots\oplus G_s$. For $\rho>0$, $\ball{\rho}(z)=\{x\in\Rset^n:\norme{x-z}{}\leq\rho\}$ where $\norme{\cdot}{}$ is the Euclidean norm in $\Rset^n$. Given a set $\Xset\subset\Rset^n$, $\convh(\Xset)$ denotes its convex hull. Function $\dist{v}{\Xset}$ denotes the distance among a vector $v$ and a set $\Xset$. The symbol $\Zero_r$ denotes a column vector in $\Rset^r$ with all elements equal to $0$. Let $v,~\bv\in\Rset^s$, the inequality $\abs{v}\leq \bar v$, component-wise means $\abs{v_i}\leq\bar v_i$, $i=1:s$.
          \begin{defi}[RCI set]
            \label{def:RCI}
            Consider the discrete-time linear system $x(t+1)=Ax(t)+Bu(t)+w(t)$, with $x(t)\in\Rset^n$, $u(t)\in\Rset^m$, $w(t)\in\Rset^n$ and subject to constraints $u(t)\in\Uset\subseteq\Rset^m$ and $w(t)\in\Wset\subset\Rset^n$. The set $\Xset\subseteq\Rset^n$ is an RCI set with respect to $w(t)\in\Wset$, if $\forall x(t)\in\Xset$ there exists $u(t)\in\Uset$ such that $x(t+1)\in\Xset$, $\forall w(t)\in\Wset$.
          \end{defi}

     \section{System definition}
          \label{sec:systemdef}
          Consider a class of discrete-time nonlinear LSSs composed of $M$ subsystems, using two different decompositions of the system structural graph (see Figure \ref{fig:decomposition}).
          \begin{figure}[!htb]
            \centering
            \includegraphics[scale=0.35]{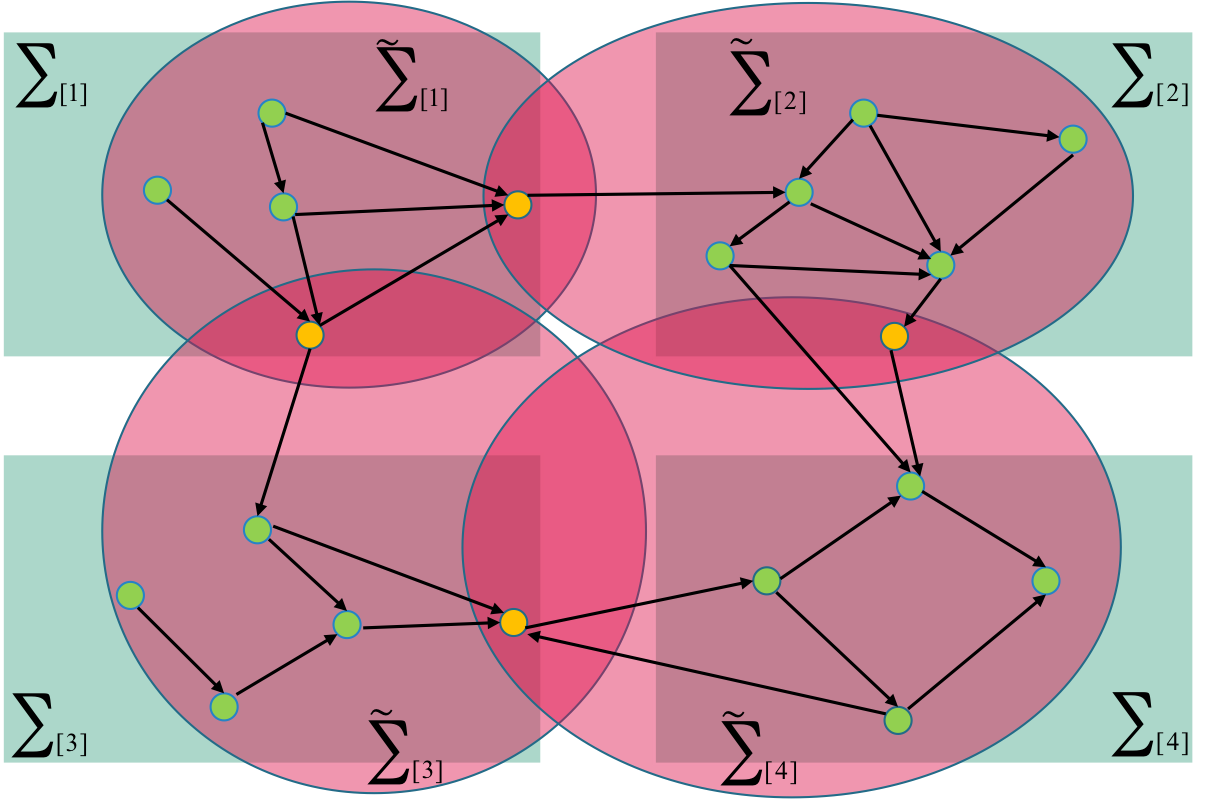}
            \caption{The two different decompositions of the LSS structural graph: the non-overlapping subsystems of the control architecture in blue and the overlapping subsystems of the fault diagnosis framework in red. The small circles represent the state and input variables; the yellow ones are the shared state variables. }
            \label{fig:decomposition}
          \end{figure}
          The control framework considers a model described by the following dynamics 
          \begin{equation}
            \label{eq:subsystem}
            \subss\Sigma i:\quad\subss \xp i=A_{ii}\subss x i+B_i [g_i(\subss x i,\subss \psi i)\subss u i + h_i(\subss x i,\subss \psi i)] + w_i(\subss\psi i)
          \end{equation}
          where $\subss x i\in\Rset^{n_i}$, $\subss u i\in\Rset^{m_i}$, $i\in\MM=\{1,\ldots,M\}$, are the local state and input, respectively, at time $t$ and $\subss\xp i$ stands for $\subss x i$ at time $t+1$. The $k$-th component of vector $\subss x i$ is specified by $\subss x {i,k}$. A similar notation is used for input and output variables. The vector of interconnection variables $\subss \psi i\in\Rset^{p_i}$ collects the states $\{\subss x j\}_{j\in\NN_i}$ that influence the dynamics of $\subss x i$, where $\NN_i$ is the set of parents of subsystem $i$ defined as $\NN_i=\{j\in\MM: \frac{\partial\subss\xp i}{\partial\subss x j}\neq \Zero_{n_i}, i\neq j\}$. We also define $\FF_{i}=\{k:i\in\NN_k\}$ as the set of children of $\Sigma_{[i]}$. $A_{ii}\in\Rset^{n_i\times n_i}$, $B_i\in\Rset^{n_i\times m_i}$, $i\in\MM$, $g_i(\cdot):\Rset^{n_i}\times\Rset^{p_i}\rightarrow\Rset$ and $h_i(\cdot):\Rset^{n_i}\times\Rset^{p_i}\rightarrow\Rset^{m_i}$, represent linear and possibly non-linear nominal dynamics. Nonlinear dynamics can also include known relationships with parent subsystems by means of the interconnection variables. The considered class of non-linear functions is general: the only constraint is the matched dependence on the control input. Instead, $w_i(\cdot):\Rset^{p_i}\rightarrow\Rset^{n_i}$ represents the unknown possibly nonlinear coupling among subsystems and includes also modeling uncertainties. We assume that the state vector is completely measurable. On the other hand, the distributed FD architecture monitors a state vector $\subss \tx i$ which is extended with respect to the controlled one, since in addition to $\subss x i$ it includes some variables $\subss x {j,s}$, $j\in\NN_i$ that it shares with parent subsystems. We call \emph{shared} variables of subsystem $i$ both the variables belonging to parents subsystems monitored also by subsystem $i$, and the variables of subsystem $i$ monitored by children subsystems.  Moreover, the FD architecture takes into account nominal model uncertainties. Therefore, the system dynamics considered by the $i$-th local diagnoser can be described as:          
          \begin{subequations}
            \label{eq:subsystemD}
            \begin{align}
              \label{eq:subsystemstate}\subss{\tilde\Sigma} i:\quad\subss \txp i=&\tilde A_{ii}\subss \tx i+\tilde B_i [\tg_i(\subss \tx i,\subss{\tilde \psi} i)\subss u i+\tildeh_i(\subss \tx i,\subss{\tilde \psi} i)]+\tw_i(\subss{\tilde \psi} i)+\phi_i(\subss \tx i,\subss {\tilde \psi} i,\subss u i,t)\\
              \label{eq:subsystemoutput}\subss y i=&\subss \tx i+\subss \varrho i
            \end{align}
          \end{subequations}
          where $\subss \tx i\in\Rset^{\tilde n_i}$, $\subss u i\in\Rset^{m_i}$, $\subss y i\in\Rset^{\tilde n_i}$ and $\subss\varrho i\in\Rset^{\tilde n_i}$, $i\in\MM$, are the local state, input, output and unknown measurement error, respectively, for diagnosis purposes. The vector of interconnection variables $\subss{\tilde \psi} i\in\Rset^{\tilde p_i}$ collects the variables $\subss{\psi} i$ except the variables shared by parent subsystems with the $i$-th subsystem.  As a consequence the state matrix $A_{ii}$ is extended to $\tilde A_{ii}$ to describe the linear dynamics of the state $\subss \tx i$, and similarly $B_i$ and functions $\tg_i$, $\tildeh_i$ and $\tw_i$, $i\in\MM$. Instead, the function $\phi_i(\cdot):\Rset^{\tilde n_i}\times\Rset^{\tilde p_i}\times\Rset^{m_i}\times\Rset\rightarrow\Rset^{\tilde n_i}$ represents the fault-function, capturing deviations of the dynamics of $\tilde \Sigma_i$ from the nominal healthy dynamics. Note that $\subss\tx i$ and $\subss{\tilde\psi}i$ are defined in a way such that computing the lhs of \eqref{eq:subsystemD} requires at most information from subsystems $\subss{\Sigma}j$, $j\in\NN_i$. In other words only transmission of information from parent to child subsystems is required.\\
          We consider the following assumptions.
          \begin{assum}
            \label{ass:standardAssumCtrl}
            \begin{enumerate}[(I)]
            \item \label{ass:localcontrollability}The pair $(A_{ii},B_i)$ is stabilizable, $\forall i\in\MM$.
            \item \label{ass:shapesets} Subsystems $\subss\Sigma i$, $i\in\MM$ are subject to the constraints
              \begin{equation}
                \label{eq:local_constraints_ctrl}
                \subss x i \in \Xset_i,~\subss u i \in \Uset_i,~\subss \varrho i \in \Oset_i
              \end{equation}
              where $\Xset_i$, $\Uset_i$ and $\Oset_i$ are compact, convex and contain the origin in their nonempty interior. Constraints \eqref{eq:local_constraints_ctrl} also induce suitable state constraints on $\subss{\tilde\Sigma} i$, $i\in\MM$, namely $\tXset_i$. Similarly, we denote with $\Psi_i$ (resp. $\tilde\Psi_i$) constraints induced on interconnection variables $\subss\psi i$ (resp. $\subss{\tilde\psi}i$).
            \item \label{ass:couplinglimits} Functions $w_i(\cdot)$ are bounded for all $i\in\MM$, i.e. there are bounded sets $\Wset_i\subset\Rset^{n_i}$ such that $w_i(\Psi_i)\subseteq\Wset_i$. Moreover if $\bar\Psi_i\subset\hat\Psi_i$ then $w_i(\bar\Psi_i)\subset w_i(\hat\Psi_i)$.
            \item Functions $g_i(\subss x i,\subss \psi i)$ are invertible for all $\subss x i\in\Xset_i$ and $\subss\psi i\in\Psi_i$.
            \item The measurement error $\subss \varrho i$ is bounded for all $i\in\MM$ at each time $t$, i.e. $| \subss \varrho i | \le \subss {\bar\varrho} i$ component-wise.
            \end{enumerate}           
          \end{assum}
          
          Now, let us provide a formal characterization of the system's decomposition already described in qualitative terms.

          \begin{defi}[\cite{Lunze1992}]
            A decomposition of the LSS into subsystems $\subss\Sigma i$, $i\in\MM$ is said \emph{non-overlapping} if no state variables are shared between subsystems. 
            Otherwise, the decomposition is termed \emph{overlapping}.
          \end{defi}
          
          In this section, we have introduced the models and the two different decompositions of the LSS we are going to consider. For what concerns the control architecture, a \emph{non-overlapping} decomposition is defined, so that each state component is controlled by only one local controller. On the other hand, an \emph{overlapping} decomposition is proposed for the FD framework, which implies that the shared state variables may be monitored by more than one local diagnosers. In the following sections, we explain how to design a control and a FD architectures suitable for a PnP framework.

     \section{Nonlinear tube-based distributed MPC}
          \label{sec:nonlinearsystube}
          
          In this section, we illustrate the proposed distributed tube-based MPC controller. We design the controller so that it is able to guarantee stability of the LSS interconnected subsystems both during the healthy behaviour (when no faults are acting on the LSS) and during the reconfiguration process (when a faulty subsystem is detected and subsequently unplugged). More specifically, we derive the DiMPC controller such that it preserves overall feasibility and stability even when a subsystem is disconnected.\\
          Concerning the control architecture, we consider a non-overlapping decomposition of the LSS. Note that, in order to design the local controllers, the model in \eqref{eq:subsystem} is used where $w_i(\cdot)$ represents coupling terms only. In the following, we propose a distributed controller that can be designed in a PnP fashion by treating parent subsystems as bounded disturbances. To this purpose, as in \cite{Rakovic2006}, we define a nominal model for each subsystem
          \begin{equation}
            \label{eq:NLnominalsubsystem}
            \subss{\hat\Sigma} i:\quad\subss \hxp i=A_{ii}\subss\hx i+B_i\subss v i
          \end{equation}
          where $\subss v i$ is the input. As in \cite{Rakovic2006} our goal is to relate inputs $\subss v i$ in \eqref{eq:NLnominalsubsystem} to $\subss u i$ in \eqref{eq:subsystem} and compute sets $\Zset_i\subseteq\Rset^{n_i}$, $i\in\MM$ such that
          \begin{equation}
            \label{eq:Zinvariance}
            \subss x i(0)\in\subss\hx i(0)\oplus\Zset_i\Rightarrow\subss x i(t)\in\subss\hx i(t)\oplus\Zset_i,~\forall t\geq 0.
          \end{equation}
          In other terms, as in \cite{Riverso2013c} and \cite{Riverso2014a}, we want to confine $\subss x i(t)$ in a tube around $\subss\hx i(t)$ of section $\Zset_i$. Assume that if $\subss x i\in\Zset_i$ there exists $\subss u i=\bkappa_i(\subss x i):\Zset_i\rightarrow\Uset_{i}$ such that $\subss\xp i\in\Zset_i$, $\forall\subss x j\in\Xset_j$, $j\in\NN_i$. Therefore if $\subss x i\in\subss\hx i\oplus\Zset_i$ and the controller
          \begin{equation}
            \label{eq:NLtubecontrol}
            \begin{aligned}
              \subss{\CC}i:\quad\subss u i=g_i(\subss x i,\subss\psi i)^{-1}[h_i(\subss x i,\subss\psi i)+\subss v i+\bkappa_i(\subss x i-\subss\bx i)]
            \end{aligned}
          \end{equation}
          is used, then, for all $\subss v i$, we have $\subss\xp i\in\subss\hxp i\oplus\Zset_i$. Controller $\subss{\CC}i$ is based on the well-known idea of ``canceling" the nonlinearities in the state equations. This is possible because in \eqref{eq:subsystem} the nonlinear terms are ``matched'', i.e. they can be directly modified through the control input $\subss u i$ \cite{Pappas1995}. 
          \begin{rem}
           We highlight that the proposed controller can be easily generalized to robust PnP MPC controllers when $w_i(\cdot)$ represents both coupling terms and model uncertainties. We refer the interested reader to Chapter 7 of \cite{Riverso2014} where robustness has been studied for linear LSSs.
          \end{rem}
          We note that controller $\subss{\CC}i$ is {\em distributed} since it depends on the state variables of parent subsystems by means of the interconnection variables. Following \cite{Rakovic2006}, the next goal is to compute tightened constraints $\hXset_i\subseteq\Xset_i$ and $\Vset_i\subseteq\Uset_i$ in order to guarantee that 
          $$\subss\hx i\in\hXset_i\mbox{ and }\subss v i\in\Vset_i \imply \subss \xp i\in\Xset_i\mbox{ and }\subss u i\in\Uset_i,$$
          at all time instants. Tightened state constraints must satisfy the following inclusions
          \begin{subequations}
            \label{eq:tightconstraint}
            \begin{align}
              \label{eq:tightstateconstraint}&\hXset_i\oplus\Zset_i\subseteq\Xset_i\\
              \label{eq:tightinputconstraint}&\Gset_i\left( \Hset_i\oplus\Vset_i \oplus \Uset_{z_i} \right) \subseteq \Uset_i
            \end{align}
          \end{subequations}
          where $\Gset_i=g_i(\Xset_i,\Psi_i)^{-1}$ and $\Hset_i=h_i(\Xset_i,\Psi_i)$. Obviously, as in nonlinear tube-based MPC theory, the evaluation of sets $\Gset_i$ and $\Hset_i$ can be very challenging. Estimates of these sets can be obtained using methods of reachability analysis for nonlinear systems, as those discussed in \cite{Raimondo2012}. Therefore, since we want to stabilize the nominal subsystems \eqref{eq:NLnominalsubsystem} and to guarantee satisfaction of tightened state constraints, we need to solve online the following {\em local} MPC problem $\Pset_i^N(\subss x i(t))$:
          \begin{subequations}
            \label{eq:decMPCProblem}
            \begin{align}
              &\label{eq:costMPCProblem}\min_{\substack{\subss\hx i(0)\\\subss v i(0:N_i-1)}}\sum_{k=0}^{N_i-1}\ell_i(\subss\hx i(k),\subss v i(k))+V_{f_i}(\subss\hx i(N_i))
            \end{align}
            \begin{align}
              &\label{eq:inZproblem}\subss x i(t)-\subss \hx i(0)\in\Zset_i\\
              &\label{eq:dynproblem}\subss \hx i(k+1)=A_{ii}\subss \hx i(k)+B_i\subss v i(k)& k\in 0:N_i-1 \\
              &\label{eq:inhXVproblem} \subss \hx i(k)\in\hXset_i,~ \subss v i(k)\in\Vset_i & k\in 0:N_i-1\\ 
              &\label{eq:inTerminalSet}\subss \hx i(N_i)\in{{\hat{\Xset}}}_{f_i}
            \end{align}
          \end{subequations}
          In \eqref{eq:decMPCProblem}, $N_i>0$ is the control horizon, $\ell_i(\cdot):\Rset^{n_i\times m_i}\rightarrow\Rset_{0+}$ is the stage cost, $V_{f_i}(\cdot):\Rset^{n_i}\rightarrow\Rset_{0+}$ is the final cost and $\hXset_{f_i}$ is the terminal set. Furthermore, following \cite{Rakovic2006}, in \eqref{eq:NLtubecontrol} we set
          \begin{equation}
            \label{eq:def_kappa_eta}
            \begin{aligned}
              \subss v i(t) = \subss v i(0|t),\qquad\subss\bx i (t) = \subss\hx i(0|t)
            \end{aligned}
          \end{equation}
          where $\subss v i(0|t)$ and $\subss\hx i(0|t)$ are optimal values of the variables $\subss v i(0)$ and $\subss\hx i(0)$ in the MPC-$i$ problem \eqref{eq:decMPCProblem}. Note that in \eqref{eq:def_kappa_eta} we defined the variable $\subss\bx i$ depending on the nominal state $\subss\hx i$, i.e. the state of the dynamics of the subsystem $\subss{\Sigma}i$ without coupling terms. Note also that the re-definition of $\subss\bx i$ as in \eqref{eq:def_kappa_eta} is at the core of the tube-MPC scheme proposed in \cite{Rakovic2006}. Algorithm \ref{alg:pnpcontrollers} summarizes the steps needed for computing function $\bkappa_i(\cdot)$ in \eqref{eq:NLtubecontrol}, sets $\Zset_i$, $\Uset_{z_i}$, $\hXset_i$, $\Vset_i$, $\hXset_{f_i}$ and functions $\ell_i(\cdot)$ and $V_{f_i}(\cdot)$.

          \begin{algorithm}[!htb]
            \caption{Design of controller $\subss\CC i$ for subsystem $\subss \Sigma i$}
            \label{alg:pnpcontrollers}
            \textbf{Input}: $A_{ii}$, $B_i$, $\Xset_i$, $\Uset_i$, $g_i(\cdot)$, $h_i(\cdot)$, $w_i(\cdot)$, $\NN_i$\\
            \textbf{Output}: controller $\subss{\CC}i$\\
            \begin{enumerate}[(I)]
            \item Send sets $\Xset_i$ to child subsystems $j\in\FF_i$
            \item Receive sets $\Xset_j$ from parent subsystems $j\in\NN_i$
            \item \label{enu:AssOmegai}Compute the set 
              \begin{equation}
                \label{eq:ch9:disturbanceControlSet}
                \begin{aligned}
                  \Wset_i&=w_i(\Psi_i)
                \end{aligned}
              \end{equation}
              and choose $\bZset_i^0$ such that $\Xset_i\supseteq\bZset_i^0\supseteq\Wset_i\oplus\ball{\omega_i}(0)$ for a sufficiently small $\omega_i>0$. If $\bZset_i^0$ does not exist, then \textbf{stop} (the controller $\subss{\CC}i$ cannot be designed)
            \item\label{enu:ch9:rciAlg} Check the LP feasibility condition in Step (ii) of Algorithm 1 in \cite{Riverso2014a}. If it is not verified, then \textbf{stop} (the controller $\subss{\CC}i$ cannot be designed)
            \item\label{enu:ch9:hXVsetAlg} Execute Steps (iii) and (iv) of Algorithm 1 in \cite{Riverso2014a}. They provide the MPC-$i$ problem and the function $\bkappa_i(\cdot)$ defined as in (25) in \cite{Riverso2014a}
            \end{enumerate}
          \end{algorithm}

          Steps (\ref{enu:ch9:rciAlg}) and (\ref{enu:ch9:hXVsetAlg}) of Algorithm \ref{alg:pnpcontrollers}, that provide constraints in \eqref{eq:tightconstraint}, are the most computationally expensive because involve Minkowski sums and differences of polytopic sets. The interested reader is referred to Sections 3.1-3.3 in \cite{Riverso2014a}, where we show how to avoid burdensome computations exploiting results from \cite{Rakovic2010} and how to compute a suitable function $\bar \kappa _i$ in \eqref{eq:NLtubecontrol} through LP. We also highlight that Step (\ref{enu:ch9:rciAlg}) is the core of the algorithm: by checking the LP feasibility condition in Step (ii) of Algorithm 1 in \cite{Riverso2014a}, we are able to verify if there exists a set $\Zset_i$ guaranteeing \eqref{eq:Zinvariance}. This is possible using a suitable parameterization of the RCI set $\Zset_i$, as proposed in \cite{Rakovic2010}.

          Next, we give the main results on stability and constraints satisfaction for the network of subsystems controlled by distributed controllers $\subss{\CC}i$.
          \begin{thm}
            \label{thm:ch9:mainclosedloop}
            Let Assumption~\ref{ass:standardAssumCtrl} hold. Assume state-feedback controllers $\subss\CC i$ are computed using Algorithm \ref{alg:pnpcontrollers} and define $\mbf x(t) = (\subss x 1,\ldots,\subss x M)$. Let $\Xset_i^N=\{\subss s i\in\Xset_i:~\mbox{\eqref{eq:decMPCProblem} is feasible for}~\subss x i(t)=\subss s i\}$ be the feasibility region for the MPC-$i$ problem and $\Xset^N=\prod_{i\in\MM}\Xset^N_i$. Then, the origin of the closed-loop system is asymptotically stable. Moreover, $\Xset^N$ is a region of attraction for the origin and $\mbf x (0)\in\Xset^N$ guarantees state and input constraints are fulfilled at all time instants.
          \end{thm}
          \begin{proof} 
            The proof of Theorem \ref{thm:ch9:mainclosedloop} is given in Appendix \ref{sec:prooftheoremctrl}.
          \end{proof}
          \begin{rem}
            Notice that Algorithm \ref{alg:pnpcontrollers} provides a decentralized procedure for designing distributed PnP regulators and that it can be executed in parallel for all subsystems. Therefore, as shown in \cite{Riverso2013c,Riverso2014a} and as we will see jointly with the FD architecture presented in Sections \ref{sec:reconfiguration} and \ref{sec:plugin}, plug-in or unplugging operations involve only the update of a limited number of controllers. Differently from \cite{Riverso2013c} and \cite{Riverso2014a} (where only linear subsystems have been considered), we highlight that the proposed regulator allows to control subsystems described by matched nonlinearities and nonlinear couplings with parents.
          \end{rem}

     \section{The Fault Detection Architecture}
          \label{sec:fault_architecture}
          In this section, we design a distributed FD architecture for the considered PnP framework. Each subsystem is equipped with a local diagnoser. According to the classical model-based FD approach, an estimate $\subss {\hat {\tx}} i$ of the local state variables is defined; the estimation error $\subss \epsilon i \triangleq \subss y i - \subss {\hat {\tx}} i$ is compared component-wise with a suitable time-varying detection threshold $\subss{\bar{\epsilon}}{i}\in\Rset_+^{\tilde n_i}$, hence obtaining a local fault decision classifying the status of the subsystem either as healthy or faulty. If the residual crosses the threshold, we can conclude that a fault has occurred. The condition $\abs{\subss{\epsilon}{i,k}(t)}\leq\subss{\bar\epsilon}{i,k}(t),\forall k=1:\tilde{n}_i$ is a necessary (but generally not sufficient) condition for the hypothesis $\HH_i:\mbox{ ``Subsystem } \subss{\tilde{\Sigma}} i\mbox{ is healthy''}$. If the condition is violated at some time instant, then the hypothesis $\HH_i$ is falsified.
          In the PnP framework, the diagnosers are designed so to guarantee the absence of false alarms and the convergence of the estimator error both during healthy conditions and during the reconfiguration process: the healthy subsystems diagnosers have to continue to work properly also when the faulty subsystem(s) is (are) unplugged and then plugged-in after problem solution.
          \subsection{The Fault Detection Estimator}
               \label{sec:fdae}
               For detection purposes, each subsystem is equipped with a local nonlinear estimator, based on the local model $\subss{\tilde\Sigma} i$ in \eqref{eq:subsystemD}. 
               The $k$-th non-shared state variable of $\subss{\tilde\Sigma}i$ can be estimated as
               \begin{equation*}
                 \subss {\hat {\tilde{x}}^+} {i,k} = \lambda(\subss {\hat {\tx}} {i,k}-\subss y {i,k})+\tilde A_{ii,k}\subss y i + \tilde B_{i,k}[\tilde g_i(\subss y i,\subss{z} i)\subss u i +\tilde h_i(\subss y i,\subss{z} i)],
               \end{equation*}
               where the filter parameter is chosen in the interval $0<\lambda<1$ in order to guarantee convergence properties, $\subss{z} i=\subss{\tilde\psi} i+\subss\theta i$ is the vector of measured interconnection variables available for diagnosis, $\subss\theta i$ collects the involved measurement error $\subss\varrho j$, $j\in\NN_i$, $\tilde A_{ii,k}$ and $\tilde B_{i,k}$ are the $k$-th row of matrices $\tilde A_{ii}$ and $\tilde B_i$, respectively. Using shared variable $\subss \tx {i,k_i}=\subss \tx {j,k_j}$, where $k_i$ and $k_j$ are the $k_i$-th and $k_j$-th components of vectors $\subss \tx i$ and $\subss \tx j$, respectively, we can take advantage of the redundancy by using a kind of deterministic consensus protocol (see \cite{Ferrari2012,Boem2011}). In the following, $\Sset^k$ is the set of subsystems $\subss{\tilde\Sigma}i$ sharing a given state variable $k$ of the LSS. The estimates of shared variables are provided by
               \begin{multline}
                 \label{eq:shared_state_dyn}
                 \subss {\hat {\tilde{x}}^+} {i,k_i} = \lambda(\subss {\hat {\tx}} {i,k_i}-\subss y {i,k_i})+ \sum_{j\in\Sset^k} W_{i,j}^k\left[\subss{\hat {\tx}}{j,k_j}-\subss{\hat {\tx}}{i,k_i}+\tilde A_{jj,k_j}\subss y j + \tilde B_{j,k_j}[\tilde g_j(\subss y j,\subss{z} j)\subss u j +\tilde h_j(\subss y j,\subss{z} j)]\right]
                 \end{multline}
               where $W_{i,j}^k$ are the components of a row-stochastic matrix $W^k$, which will be defined in Subsection~\ref{sec:consensus}, designed to allow plugging-in and unplugging operations. By now, notice that $W^k$ collects the consensus weights used by $\subss{\tilde\Sigma} i$ to weight the terms communicated by $\subss{\tilde\Sigma} j$, with $j\in\Sset^k$. In fact, as regards variables estimation, each subsystem communicates with parents and children subsystems sharing that variable. We also note that \eqref{eq:shared_state_dyn} holds also for the case of non-shared variables, since, in this case, $\Sset^k=\{i\}$, and  $W_{i,i}^k=1$ by definition. In the following, for the sake of simplicity, we omit the subscript of the shared component index $k$ $\subss \tx {i,k}$ instead of $\subss \tx {i,k_i}$.
          
          \subsection{The detection threshold}
               \label{sec:threshold}
               In order to define an appropriate threshold for FD, we analyze the dynamics of the local diagnoser estimation error when the subsystem is healthy. Defining $W^k$ such that $\sum_{j\in\Sset^k} W_{i,j}^k=1$ and since for shared variables $\forall i, j\in\Sset^k$ it holds
               \begin{equation*}
                 \tilde A_{ii,k}\subss \tx i + \tilde B_{i,k}[\tilde g_i(\subss \tx i,\subss{\tilde\psi} i)\subss u i +\tilde h_i(\subss y i,\subss{z} i)]=\tilde A_{jj,k}\subss \tx j + \tilde B_{j,k}[\tilde g_j(\subss \tx j,\subss{\tilde\psi} j)\subss u j +\tilde h_j(\subss y j,\subss{z} j)],
               \end{equation*}
               the $k$-th state estimation error dynamics is given by
               \begin{multline*}
                 \subss{\epsilon^+}{i,k} = \sum_{j\in\Sset^k} W_{i,j}^k\left[ \lambda \subss{\epsilon}{j,k}- \tilde A_{jj,k}\subss\varrho j + w_{j,k}(\subss {\tilde\psi} j) + \tilde B_{j,k}(\Delta \tg_{j,k}\subss u j+\Delta \tildeh_{j,k})-\lambda\subss\varrho{j,k} \right] + \lambda\subss\varrho{i,k} + \subss{\varrho^+}{i,k} \, ,
               \end{multline*}
               where $\Delta \tg_{j,k} \triangleq \tilde g_{j,k}(\subss \tx j,\subss {\tilde\psi} j)-\tilde g_{j,k}(\subss y j,\subss{z}j)$ and $\Delta \tildeh_{j,k} \triangleq \tilde h_{j,k}(\subss \tx j,\subss {\tilde\psi} j)-\tilde h_{j,k}(\subss y j,\subss{z}j)$.

               As in \cite{Ferrari2012}, using the triangular inequality, we can bound the estimation error, guaranteeing no false-positive alarms.  By taking the absolute value of $\subss{\epsilon^+}{i,k}$ component-wise, we get
               \begin{multline*}
                 \abs{\subss{\epsilon^+}{i,k}} \leq \sum_{j\in\Sset^k} W_{i,j}^k\left[ \lambda \abs{\subss{\epsilon}{j,k}} + \abs{\tilde A_{jj,k}\subss\varrho j} + \lambda\abs{\subss\varrho{j,k}} + \abs{\tilde B_{j,k}(\Delta \tg_{j,k}\subss u j+\Delta \tildeh_{j,k})}+\abs{w_{j,k}(\subss {\tilde\psi} j)} \right]\\ + \lambda\abs{\subss\varrho{i,k}} + \abs{\subss{\varrho^+}{i,k}} \, .
               \end{multline*}

               Therefore, we define the following time-varying threshold $\subss{\bar\epsilon}{i,k}$ that {\em can be computed in a distributed way} as
               \begin{multline}
                 \label{eq:threshold}
                 \subss{\bar\epsilon^+}{i,k} = \sum_{j\in\Sset^k} W_{i,j}^k\left[ \lambda \subss{\bar\epsilon}{j,k} + \left|\tilde A_{jj,k}\right|\subss{\bar\varrho}j + \bar w_{j,k}(\subss {z} j) + \left|\tilde B_{j,k}\right| {(\Delta \bar g_{j}}\left|\subss u j\right|+\Delta\bar h_{j})+\lambda\subss{\bar\varrho}{j,k} \right] + \lambda\subss{\bar\varrho}{i,k} + \subss{\bar\varrho^+}{i,k}
               \end{multline}
               where ${\Delta \bar g_{j}}=\max_{\subss \tx j\in\tXset_j,\subss {\tilde\psi} j\in\tilde\Psi_j} \abs{\Delta \tilde g_j(t)}$ and ${\Delta \bar h_{j}}=\max_{\subss \tx j\in\tXset_j,\subss {\tilde\psi} j\in\tilde\Psi_j} \norme{\Delta \tilde h_j(t)}{\infty}$.  It is worth noting that Assumption~\ref{ass:standardAssumCtrl} implies that the state and input variables are bounded; hence all quantities in~\eqref{eq:threshold} are bounded as well; moreover, it is possible to define $\forall i,\, k$ at each time step a bound $\bar w_{i,k}$, so that $\left|w_{i,k}(\subss {z} i)\right|\leq\bar w_{i,k}(\subss {z} i)$; $\subss{\bar\varrho}{i,k}$ is defined in Assumption~\ref{ass:standardAssumCtrl}. The threshold dynamics \eqref{eq:threshold} can be initialized with $\subss{\bar{\epsilon}}{i,k}(0)=\subss{\bar{\varrho}}{i,k}(0)$.

               \begin{rem}
                 \label{rem:comm}
                 For FD purposes, the communication between subsystems is limited. It is not necessary, in general, that each diagnoser knows the complete model of parent subsystems. Instead, in the shared case \eqref{eq:shared_state_dyn}, it is sufficient that each subsystem $\subss{\tilde\Sigma} i$ sends to subsystems in $\Sset^k$ only a limited number of variables (the interconnection variables and the consensus terms for estimates and thresholds), locally computed.
               \end{rem} 

               The threshold in \eqref{eq:threshold} guarantees the absence of false-positive alarms before the occurrence of the fault caused by the uncertainties. On the other hand, this is a conservative result since it does not allow to detect faults whose magnitude is lower than the uncertainties magnitude in the system dynamics. This issue is formalized in the fault detectability section (Section \ref{sec:faulty}), where we consider also the issue that the fault may be hidden by the control action.
          
          \subsection{The consensus matrix}
               \label{sec:consensus}
               In this subsection we explain how to properly define the consensus matrix in order to allow for PnP operations. Consensus is applied to the shared variables, i.e. state variables representing the interconnection between two or more subsystems. For PnP capabilities, we use a time-varying weighting matrix $W^k$ whose dimension is equal to the maximum number of subsystems that can be plugged in sharing that variable. This is not a restrictive assumption since it is possible to choose a dimension as large as wanted. Each row can have non null elements only on correspondence of connected (plugged-in) subsystems. In the case that, at a given time, the variable is not shared (and hence at most one subsystem is using it) the only non-null weight is the one corresponding to the considered subsystem (this does not affect the convergence of the FD estimator as illustrated in Subsection~\ref{sec:estim:conv}).
               
               Indeed, the introduction of the proposed time-varying consensus matrix is advantageous from a second perspective. Since the proposed threshold is conservative, it is important to choose it as small as possible. Therefore, in the case of shared variables, similarly as in \cite{Boem2013}, we design a time-varying consensus-weighting matrix $W^k$ able to minimize the adaptive threshold with respect to the consensus weights, by choosing the smallest threshold term from all the threshold additive terms in \eqref{eq:threshold}. In this consensus protocol, it is convenient to weight more the subsystem which has got the lowest threshold component, hence the subsystem that has lower uncertainty in its measurements and in the local model. These aims can be achieved by defining the following consensus matrix, where each $(i,j)$-th component is computed as:
               \begin{equation}
                 \label{eq:Wdef}
                 W_{i,j}^k =
                 \begin{cases}
                   1 & \mbox{if } j = \arg\min_{j\in\Sset^k} \lambda ( \subss{\bar\epsilon}{j,k} + \subss{\bar\varrho}{j,k} )+ \left|\tilde A_{jj,k}\right|\subss{\bar\varrho}j  + \left|\tilde B_{j,k}\right| {(\Delta \bar g_{j}}\left|\subss u j\right| + \Delta \bar h_{j} )+ \bar w_{j,k}(\subss {z} j) \\
                   0 & \mbox{otherwise}
                 \end{cases}
               \end{equation}
               At each time-step each local fault-diagnoser receives estimates and consensus terms of variable $\subss\tx{i,k}$ only from the subsystems sharing it at that specific time. Then, it selects the contribution affected by ``smaller uncertainty". It is worth noting that the set $\Sset^k$ is time-varying and collects only the subsystems that share variable $k$ and that are connected to the LSS at that specific time step. As briefly discussed in Section \ref{sec:faulty}, fault-detectability may be improved by this approach.
               
               
          \subsection{Estimator convergence}
               \label{sec:estim:conv}
               
               Next, we address the convergence properties of the overall estimator before the possible occurrence of a fault, that is for $t<T_0$. Towards this end, we introduce a vector formulation of the state error equation for sake of compacting the notation, just for analysis purposes. Specifically, we introduce the extended estimation error vector $\epsilon_{k,E}$, which is a column vector collecting the estimation error vectors of the $N_k$ subsystems sharing the $k$-th state component: $\epsilon_{k,E} \triangleq \mbox{col} \left( \subss\epsilon{j,k}:~ j\in\Sset^k \right)$. Hence, the dynamics of $\epsilon_{k,E}$ can be described as:
               \begin{equation}
                 \label{eq:extended_dynamics}
                 \epsilon_{k,E}^{+}=W^k\left[\lambda\epsilon_{k,E}+\tilde A_{k,E}\varrho_{E}+\tilde B_{k,E}(\Delta \tg_{E}u_{E}+\Delta \tildeh_{E})+w_{k,E}-\lambda\varrho_{k,E}\right]+\lambda\varrho_{k,E}+\varrho_{k,E}^+,                  
               \end{equation}
               where $\varrho_{k,E}$ is a column vector, collecting the corresponding $k_j$ value of vector $\subss\varrho{j}$, i.e. $\subss\varrho{j,k_J}$, for each $j\in\Sset^k$; $\tilde A_{k,E}$ is a block matrix with $N_k$ rows and $n_E=\sum_{j=1}^{N_k}\tilde{n}_{J}$ columns, $j\in\Sset^k$, where the elements on the diagonal are the row vectors $\tilde A_{jj,k}$; $\tilde B_{k,E}$ is defined in an analogous way. Finally, $\varrho_E$, $\Delta \tg_{E}$, $\Delta \tildeh_{E}$ and $u_E$ are column vectors collecting the vectors $\subss\varrho j$, $\Delta \tg_{j}$, $\Delta \tildeh_{j}$ and $\subss u j$, with $j\in\Sset^k$, respectively, $w_{k,E}$ is defined in an analogous way. The following convergence result is now in place.

               \begin{prop}
                 System \eqref{eq:extended_dynamics}, where the consensus matrix is given by \eqref{eq:Wdef}, is a BIBO stable.
               \end{prop}

               \begin{proof}
                 The proof is carried out exploiting the one reported in \cite{Boem2013} in a purely distributed fault-diagnosis framework. Specifically, since $W^k$ is a stochastic matrix, its norm is always equal to $1$. Therefore, since $0<\lambda<1$, then also $\norme{\lambda W^k(t)}{}\leq\gamma<1$, with $0<\gamma<1$. Let us define:
                 \[U_{k,E}(t)=W^k(t)\left[\tilde A_{k,E}\varrho_{E}(t)+\tilde B_{k,E}(\Delta \tg_{E}u_{E}(t)+\Delta \tildeh_{E}(t))+w_{k,E}(t)-\lambda\varrho_{k,E}(t)\right]+\lambda\varrho_{k,E}(t)+\varrho_{k,E}(t+1).
                 \]
                 We have:
                 \begin{equation*}
                   \begin{aligned}
                     \norme{\epsilon_{k,E}(t+1)}{}&\leq\norme{\lambda W^k(t)\epsilon_{k,E}(t)}{}+\norme{U_{k,E}(t)}{}\\
                                                                &\leq\norme{\lambda W^k(t)}{}\norme{\lambda W^k(t-1)}{}\ldots\norme{\lambda W^k(0)}{}\norme{\epsilon_{k,E}(0)}{}\\
                                                                &+\sum_{j=1}^{t}\norme{\lambda W^k(t)}{}\norme{\lambda W^k(t-1)}{}\ldots\norme{\lambda W^k(j)}{}\norme{U_{k,E}(j)}{}\\
                                                                &\leq\gamma^t\norme{\epsilon_{k,E}(0)}{}+\sum_{j=1}^{t}\gamma^{t-j}\norme{U_{k,E}(j)}{}\\
                                                                &\leq \frac{1}{1-\gamma}\sup_{j\geq 1}\norme{U_{k,E}(j)}{}
                   \end{aligned}
                 \end{equation*}
                 For $t\rightarrow \infty$, the unforced system converges to zero and the series converges to a bounded value (see results in \cite{Michaletzky2002bibo}). Moreover, using results in \cite{Sichitiu2003} for unforced systems, we can state that a system $x(t + 1) = A(t)x(t)$, with $A(t)\in\convh(A_1,\dots, A_N)$, it is exponentially stable iff $\exists$ a sufficiently large integer $q$ such that $\norme{A_{i_1}~A_{i_2}\dots A_{i_q}}{} \leq \gamma < 1,~\forall (i_1,\dots,i_q)\in\left\{1,\dots,N\right\}^q$. In our case, therefore, we only need to analyze matrix $W^k(t)$. Since each row of $W^k(t)$ has all null elements except one equal to $1$, the product $W^k(t)W^k(t-1)\dots W^k(0)$ is a stochastic matrix. Hence, since $0<\lambda < 1$, we have $\norme{\lambda^t(W^k(t)W^k(t-1)\ldots W^k(0))}{}< 1$ and the hypothesis is satisfied. Finally, since all the uncertain terms are bounded, then the discrete-time system \eqref{eq:extended_dynamics} is BIBO stable.
               \end{proof}

     \section{Fault Detectability Analysis}
          \label{sec:faulty}
 In this section, we analyze fault detectability properties of the proposed FD architecture. In particular, we highlight the effects of the control input on fault detectability conditions.
         Let us now consider the  case of a faulty subsystem, that is, suppose that a fault $\phi(\cdot)$ occurs at an unknown time $t=T_{0}$ on the $k$-th state variable. In the general case of a shared variable, $\phi_{k,E}=\subss{\phi}{\ ,k} \cdot (1, \dots ,1)^T$ denoting the extended fault function vector collecting the fault functions of the subsystems sharing the $k$-th variable. After the occurrence of the fault, for $t>T_{0}$, the state estimation error dynamics is given by:
          \begin{equation*}
            \epsilon_{k,E}^{+}=W^k\left[\lambda\epsilon_{k,E}+\tilde A_{k,E}\varrho_{E}+\tilde B_{k,E}(\Delta \tg_{E}u_{E}+\Delta \tildeh_{E})+w_{k,E}-\lambda\varrho_{k,E}\right]+\lambda\varrho_{k,E}+\varrho_{k,E}^++\phi_{k,E}.
          \end{equation*}
          Then, at a time instant $t_{1}>T_{0}$, the estimation error is
          \begin{equation*}
            \begin{aligned}
              \epsilon_{k,E}(t_1) &=\sum_{h=0}^{t_{1}-1}(\lambda W^{k}(h))^{t_{1}-1-h}[-W^k(h)\tilde A_{k,E}\varrho_{E}(h)+\tilde w_{k,E}(h)+W^k(h)\tilde B_{k,E}(\Delta \tg_{E}u_{E}(h)+\Delta \tildeh_{E})\\
              &-\lambda W^k(h)\varrho_{k,E}(h)+\lambda\varrho_{k,E}(h)+\varrho_{k,E}(h+1) +\phi_{k,E}(h)]+\prod_{h=0}^{t_1-1}(\lambda W^k(h))\epsilon_{k,E}(0)
            \end{aligned}
          \end{equation*}
          Now, we derive a sufficient condition in order to characterize a class of faults that can be detected by the proposed FD scheme. In order to detect the occurrence of the fault at a certain time $t_1$, the following inequality has to be satisfied:
          \begin{equation*}
            \left|\epsilon_{k,E}(t_1)\right|>\bar{\epsilon}_{k,E}(t_1),
          \end{equation*}
          for at least one subsystem $i \in\Sset^k$. When dealing with vectors, in this paper, the inequality operator is applied component-by-component.
          Using the triangle inequality and the threshold definition \eqref{eq:threshold}, the following is implied
          \begin{equation*} 
            \left|\epsilon_{k,E}(t_{1})\right|\geq-\bar{\epsilon}_{k,E}(t_{1})+\left|\sum_{h=T_{0}}^{t_{1}-1}[\lambda^{t_{1}-1-h}\phi_{k,E}(h)]\right|.
          \end{equation*}
          Since $\phi_{k,E}$ is a vector whose components are all equal to $\phi_k=\phi_{i,k_i}=\phi_{j,k_j}$, it is easy to see that the FD condition $\left|\subss\epsilon{i,k}(t_1)\right|>\subss{\bar\epsilon}{i,k}(t_1)$ is satisfied if
          \begin{equation}
            \label{eq:condition}
            \exists t_{1}>T_{0} \, \colon \,
            \left|\sum_{h=T_{0}}^{t_{1}-1}\lambda^{t_{1}-1-h}\phi_k(h)\right|>2\subss{\bar{\epsilon}}{i,k}(t_{1})
          \end{equation}
          for at least one component $k\in \{1\ldots,\tilde{n}_i \}$, thus allowing the detection of a fault at time $t_{1}$. Condition \eqref{eq:condition} implicitly characterizes the class of faults that are detectable by the proposed FD architecture at time $t_{1}$. Moreover, thanks to the introduction of the time-varying consensus weighting matrix, the threshold on the right-hand-side of \eqref{eq:condition} is the smallest one in the set of the proposed conservative thresholds of subsystems sharing the same variable, guaranteeing no false alarms. The choice of a smaller threshold makes it easier the detectability at the general time instant $t_1$, thus we can say intuitively from \eqref{eq:condition} that the class of detectable faults at time $t_1$ is enlarged thanks to this choice.          
          
            It is worth emphasizing the influence of the control inputs on the fault detectability condition by rewriting \eqref{eq:condition} as
            \begin{multline}
              \label{eq:condition2}
              \left|\sum_{h=T_{0}}^{t_{1}-1}\lambda^{t_{1}-1-h}\phi_{k,E}(\tx_E,{\tilde \psi}_E,u_E,h)\right|>\\2 \Big (\sum_{h=0}^{t_{1}-1}(\lambda W^{k}(h))^{t_{1}-1-h}[W^k(h)\Big(\left|\tilde A_{k,E}\right|\bar{\varrho}_{E}(h) +\bar{w}_{k,E}(h)+\left|\tilde B_{k,E}\right|(\Delta \bar g_{E}\left| u_{E}(h)\right|+\Delta \bar h_{E})\\
              +\lambda \bar\varrho_{k,E}(h)\Big)+\lambda\bar\varrho_{k,E}(h)+\bar\varrho_{k,E}(h+1)]+\prod_{h=0}^{t_1-1}(\lambda W^k(h))\epsilon_{k,E}(0) \Big ) \, .
               \end{multline}
            
            Actually, the norm of the control term $u_E(t_{1}-1)$ affects the threshold on the right side of the inequality and, in particular, it may have a detrimental effect on the fault detectability by increasing the detection threshold. On the other hand, the control influences also the left part of the condition inequality, by acting on the fault function, which depends directly on $u_E(t_{1}-1)$ and, by means of $\tx_E$, it depends also on the past history of the control input.
            In order to analyze this point, it is possible to rewrite \eqref{eq:condition2} as
            
            \begin{multline}
              \label{eq:condition3}
              \left|\sum_{h=T_{0}}^{t_{1}-1}(\lambda W^{k}(h))^{t_{1}-1-h}W^k(h)\phi_{k,E}(\tx_{E},\tilde \psi_{E}, u_{E},h)\right|>\\2 \Big (\sum_{h=0}^{t_{1}-1}(\lambda W^{k}(h))^{t_{1}-1-h}[W^k(h)\left|\tilde B_{k,E}\right| (\Delta \bar g_{E}\left| u_{E}(h)\right|+\Delta \bar h_{E})]+\varsigma_E(h) \Big )
            \end{multline}
            
          where 
          \begin{multline*}
          \varsigma_E=2 \Big (\sum_{h=0}^{t_{1}-1}(\lambda W^{k}(h))^{t_{1}-1-h}[W^k(h)\Big(\left|\tilde A_{k,E}\right|\bar{\varrho}_{E}(h) +\bar{w}_{k,E}(h)+\Delta \bar h_{E})+\lambda \bar\varrho_{k,E}(h)\Big)\\+\lambda\bar\varrho_{k,E}(h)+\bar\varrho_{k,E}(h+1)]+\prod_{h=0}^{t_1-1}(\lambda W^k(h))\epsilon_{k,E}(0) \Big )
          \end{multline*}
          is the threshold part that does not depend directly on the extended control input. Therefore, it is constant w.r.t. the control input\footnote{This could be not always true since the control input could influence also the bounds of the measurement error and coupling by means of the state dynamics. However in some cases this dependence could be neglected especially when considering conservative bounds.}.
          As a consequence, the contribution of the control input to detectability properties at a certain time $t_1$ could be highlighted by deriving the vectors of functions $\left|\phi_{k,E}(x_{E},\tilde \psi_{k,E}, u_{E},h)\right|$ and $\left|\tilde B_{k,E}\right|(\Delta \bar g_{E}\left| u_{E}(h)\right|+\Delta \bar h_{E})$ w.r.t. the vector $u_E$ norm component-by-component. If it is possible to obtain the derivatives vector of the fault function we want to detect (as example, if it is possible to assume that it is a Lipschitz function w.r.t. the control input norm and to know the Lipschitz constant), then, it is possible to compare the two derivatives for each subsystem $i \in\Sset^k$. In fact, the right side term is linear w.r.t. to the norm of the control input. Intuitively, if the control input norm makes the magnitude of the fault function grow less than the threshold bounds, then the control input has a detrimental effect on detectability at time step $t_1$, since it increases the uncertainty threshold terms that hide the fault effects. On the other hand, if the control input norm makes the magnitude of the fault function grow much more than the threshold bounds, then it could be possible to take advantage of the control input effect trying to improve detectability.

     \section{Reconfiguration strategy}
          \label{sec:reconfiguration}
          In the previous sections we derived suitable control and fault detection architectures for a PnP framework. We now explain how to use them during plugging-in and unplugging operations. In this section, the reconfiguration of the LSS, in case of detection of a fault in one of the subsystems, is addressed. In general, depending on the specific application context, two distinct actions may turn out to be feasible: i) immediate ``disconnection" of the faulty subsystem or ii) continuation of the system operation in ``safety mode". As in this paper we deal with an {\em active distributed fault-tolerant} control scheme, we consider only the first scenario and, in the following, the {\em unplugging} after fault-detection and the possible {\em plug-in} after subsystem repair/replacement are addressed separately. We assume that, when the plant is started, all subsystems are healthy and governed by local controllers designed through Algorithm \ref{alg:pnpcontrollers}.
          \begin{figure}[!htb]
            \begin{center}$
              \begin{array}{cc}
                \includegraphics[scale=0.25]{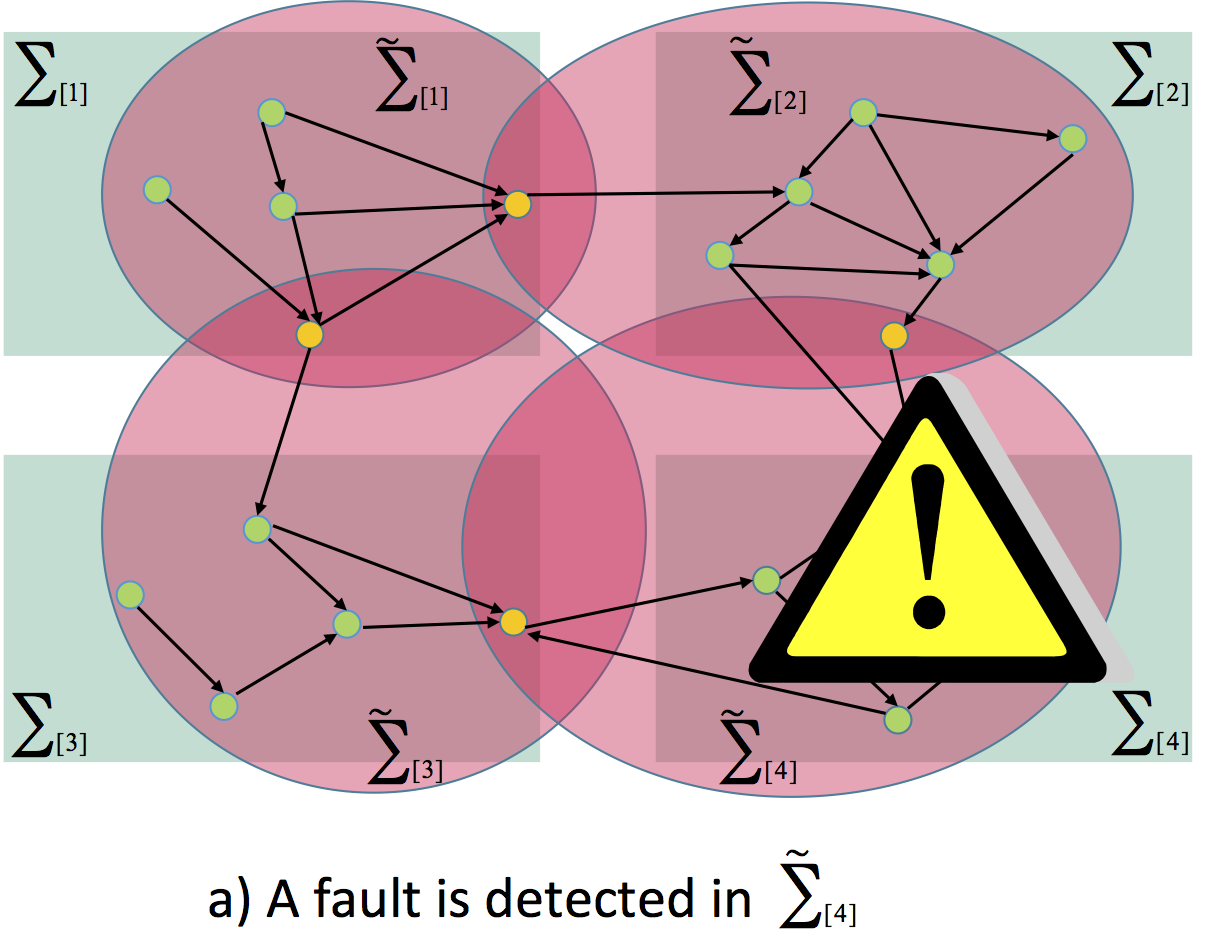}
                \includegraphics[scale=0.25]{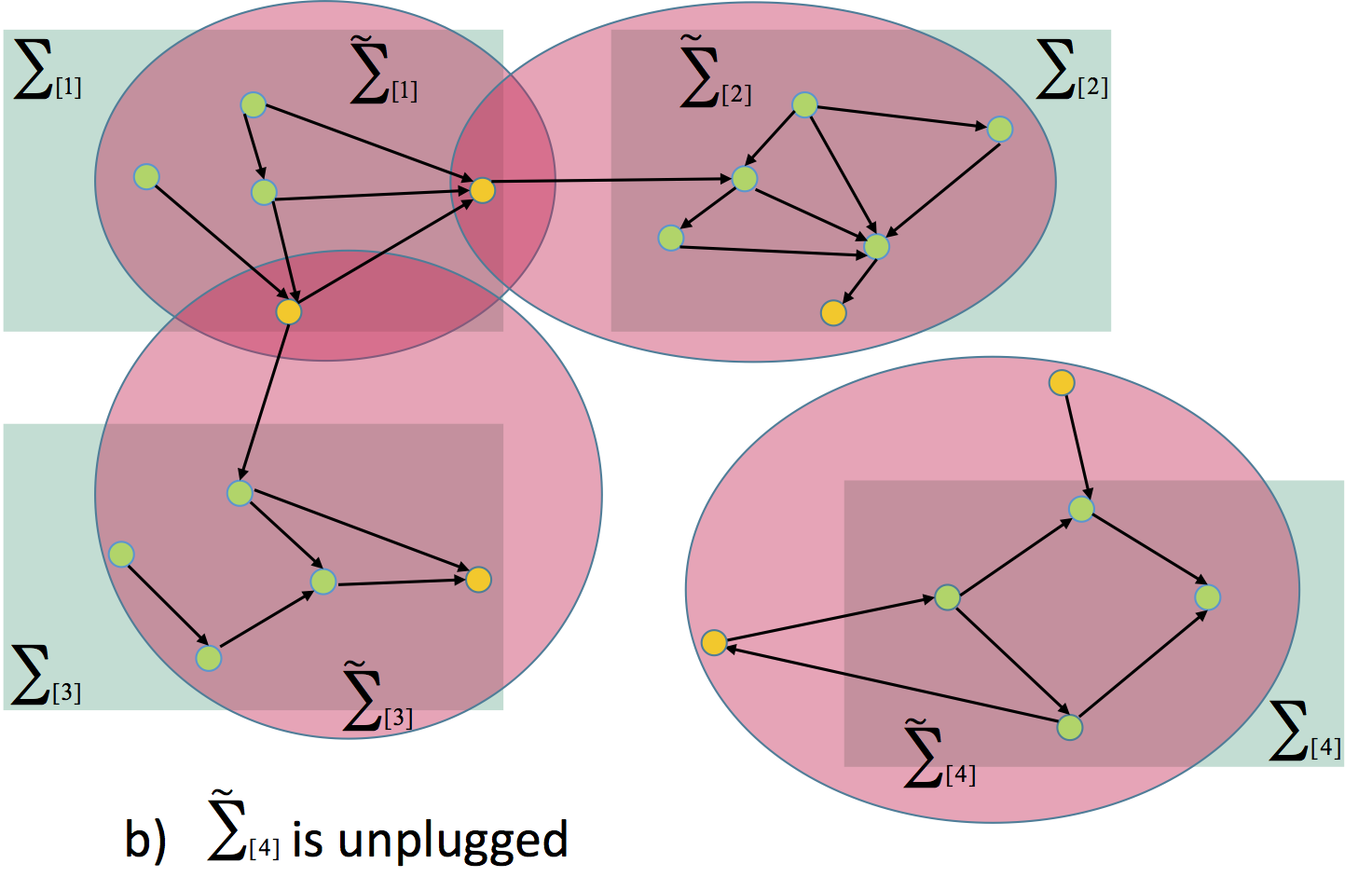}\\
                \includegraphics[scale=0.25]{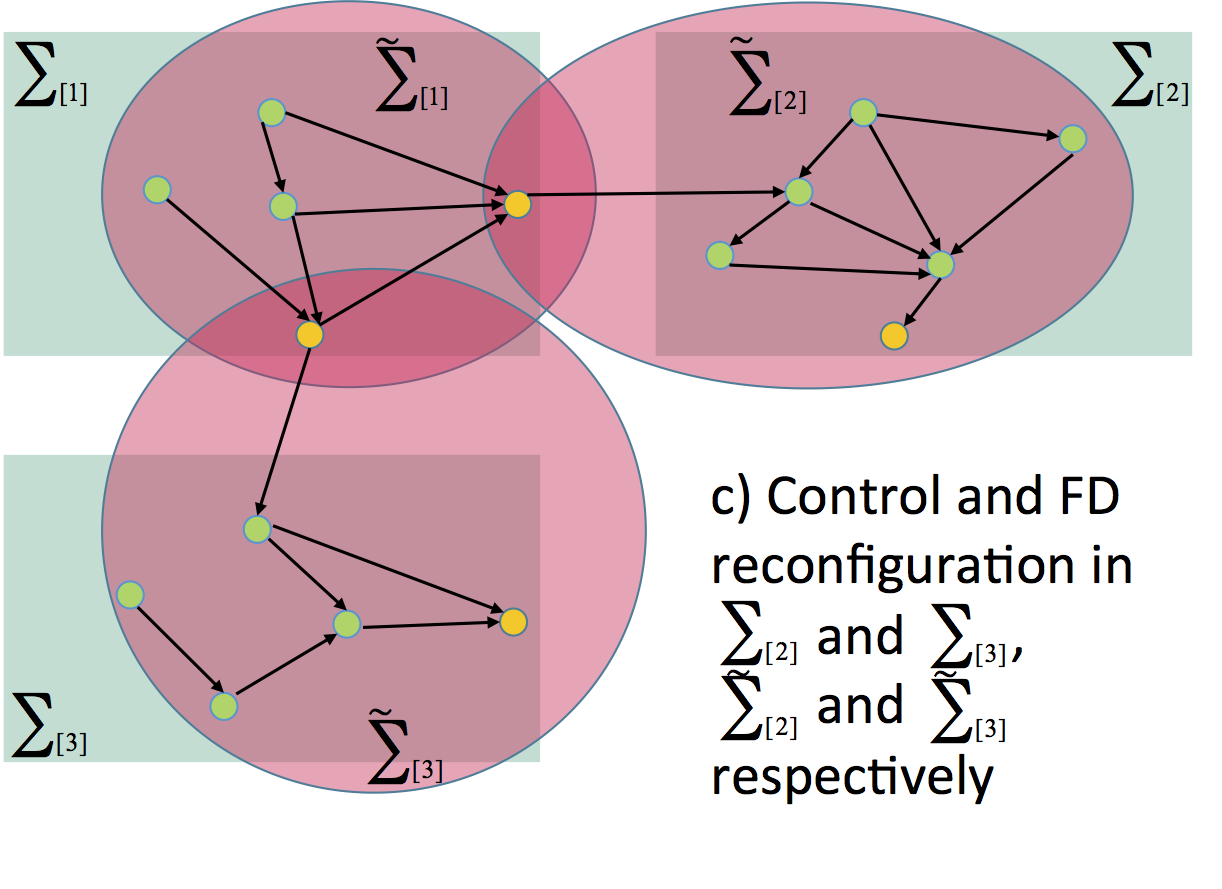}
                \includegraphics[scale=0.25]{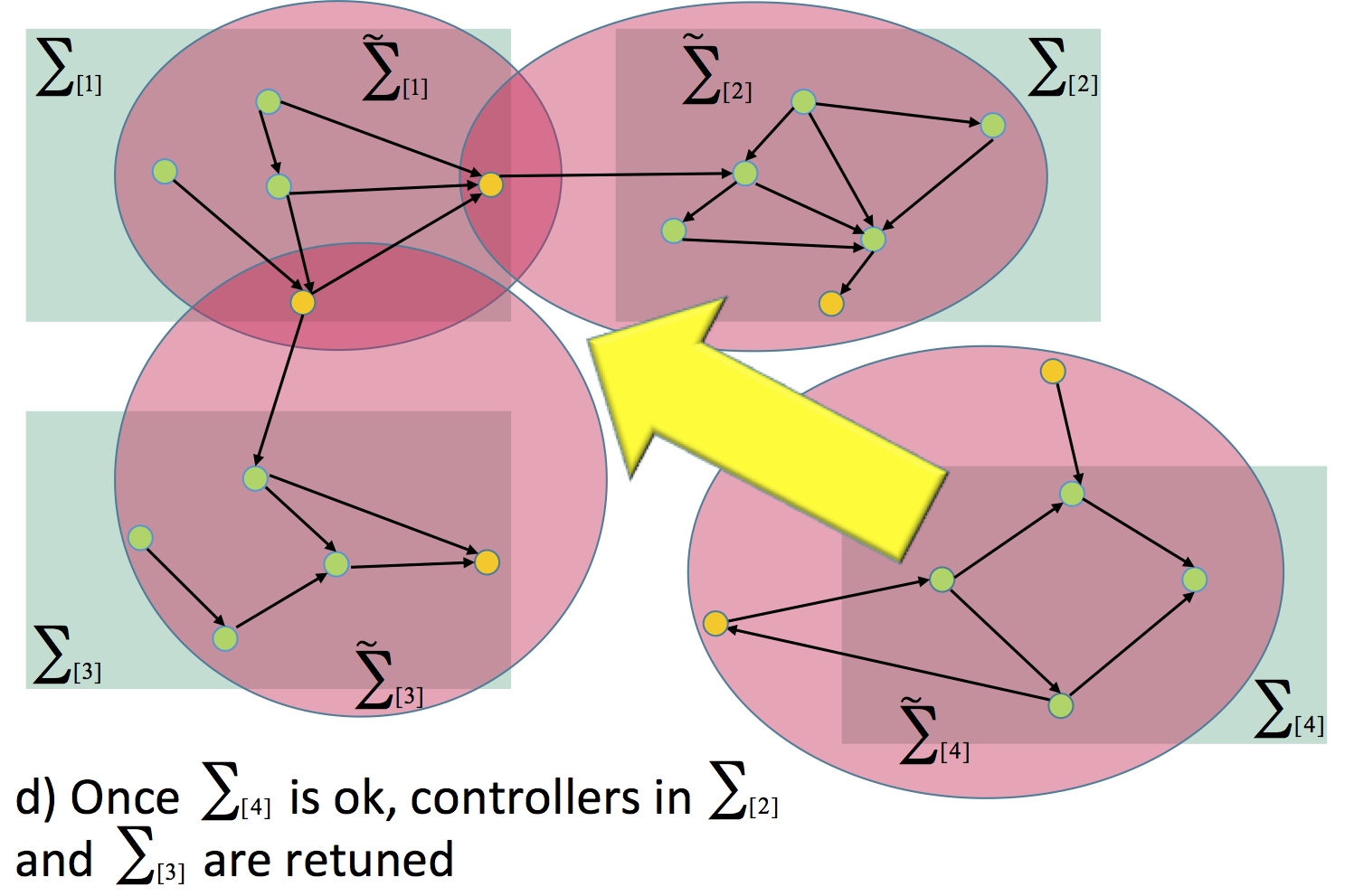}
              \end{array}$
            \end{center}
            \caption{The reconfiguration process: a), b), c), d) steps.}
            \label{reconfiguration}
          \end{figure}
          
          \subsection{Subsystem unplugging after fault detection}
               In this section, we show how to reconfigure local controllers and fault-detectors when a fault is detected in a subsystem. The proposed strategy is based on the isolation of the faulty subsystem and on the reconfiguration of controllers and fault-detectors to guarantee closed-loop stability, constraint satisfaction and monitoring of the new network with one less subsystem.

               In the following, we describe in depth the needed operations after a fault detection. Let $t=t_1$ the detection time of a fault in the $j$-th subsystem ($\subss{\tilde \Sigma} j$ in the FD architecture and $\subss{\Sigma} j$ in the control architecture), then the faulty subsystem is unplugged and the involved subsystems reconfigured.\\
               As regards the distributed FD, we need to perform the following operations.
               \begin{itemize}
               \item In the children subsystems $i \in\FF_j$, for $t\ge t_1$, the components of $\subss {\tilde\psi} i$ and $\subss z i$ related to subsystem $\subss{\tilde \Sigma} j$ become equal to $0$. Hence, for $t\ge t_1$, the interconnection variables and measurements related to subsystem $\subss{\tilde \Sigma} j$ do not influence the time-behaviour of the state estimation \eqref{eq:shared_state_dyn} and of the threshold \eqref{eq:threshold} of subsystems $\subss{\tilde \Sigma} i$.                  
               \item In the children subsystems $i \in\FF_j$, the adaptive threshold $\subss{\bar{\epsilon}}{i}$ is computed through \eqref{eq:threshold} by not considering the coupling terms related to the $j$-th subsystem when computing $\bar{w}_i$ for $t\ge t_1$.                 
               \item In the neighbouring subsystems $i$, with $i \in\FF_j$ or $i \in\NN_j$, sharing some variables with $\subss{\tilde \Sigma} j$, the weights associated with $\subss{\tilde \Sigma} j$ in the consensus matrices $W^k$ computed in \eqref{eq:Wdef} are set to zero, that is, $j\notin\Sset^k$ for $t\ge t_1$ for all the shared variables $k$.
               \end{itemize}
               
               Beyond the above changes in the local estimators embedded in the distributed FD framework as a consequence of the subsystem unplugging after the detection of a fault, the reconfiguration of the control architecture has to be addressed as well.
               Under Assumption \ref{ass:standardAssumCtrl}-(\ref{ass:couplinglimits}), for each $i\in\FF_j$, a contraction of the set $\NN_i$ takes place, since subsystem $\subss\Sigma i$ has one parent less. Then, a contraction takes place also on set $\Wset_i$ in \eqref{eq:ch9:disturbanceControlSet} and the set $\bZset_i^0$ already computed still verifies the inclusions in Step (\ref{enu:AssOmegai}) of Algorithm \ref{alg:pnpcontrollers}. Therefore, for each $i\in\FF_j$, the previous choice of $\bZset_i^0$ (made before the unplugging) still guarantees the feasibility of the LP problem in Step (\ref{enu:ch9:rciAlg}) of Algorithm \ref{alg:pnpcontrollers} which finally implies that there is no need of redesigning the controller $\subss\CC i$ to keep the overall stability. 

               In conclusion, thanks to the distributed MPC controllers and distributed fault detectors schemes we designed, the detection of a fault in a subsystem implies the isolation of the faulty subsystem and the reconfiguration of local controllers and fault detectors, at most, of parent and children subsystems. This guarantees the fault is not propagated in the network.

          \subsection{Subsystem plugging-in}
               \label{sec:plugin}
               
               The plug-in of a subsystem into the LSS interconnected structure may be needed in case of replacement of a previously unplugged subsystem the fault diagnoser in use before subsystem disconnection can be reused. Since we assumed controllers $\subss\CC i$ existed for the subsystem and its children when it was connected to the plant, this operation is always feasible as regards the control framework. For what concerns the distributed FD architecture, thanks to the way the time-varying shared variables estimator is defined, the plug-in is always feasible as well.

               \begin{rem}
                 Note that, differently from \cite{Riverso2013c,Riverso2014a}, here we do not consider the plugging-in of new subsystems but just the reconnections of subsystems after they have been repaired. Therefore, existence of controllers $\subss\CC i$ when all subsystems are healthy guarantees that after a plugging-in or unplugging operation in real-time
                 \begin{itemize}
                 \item constraints on the input and states of all subsystems are still fulfilled;
                 \item the new mode of operation of the whole plant is asymptotically stable (Theorem \ref{thm:ch9:mainclosedloop}).
                 \end{itemize}
               \end{rem}

               However, as well known in the hybrid system literature \cite{Hespanha1999}, frequent and persistent switching between different modes of operation could compromise asymptotic stability of the whole plant. A remedy could be assuming a minimal dwell-time between consecutive switches \cite{Hespanha1999} although this issue deserves further investigations.

     \section{Examples}
         \label{sec:simulationExample}

         \subsection{Coupled van der Pol oscillators} 
               \label{sec:vanDerPol}
               In this example, we apply the proposed methodologies to a ring of coupled vdPOs as in Figure \ref{fig:vdPOsRing}. They can be used to model many oscillating systems in a wide area of applications, including biological rhythms, heartbeat, chemical oscillations, circadian rhythms \cite{Barron2008}. 
               \begin{figure}[!htb]
                 \centering
                 \includegraphics[scale=0.25]{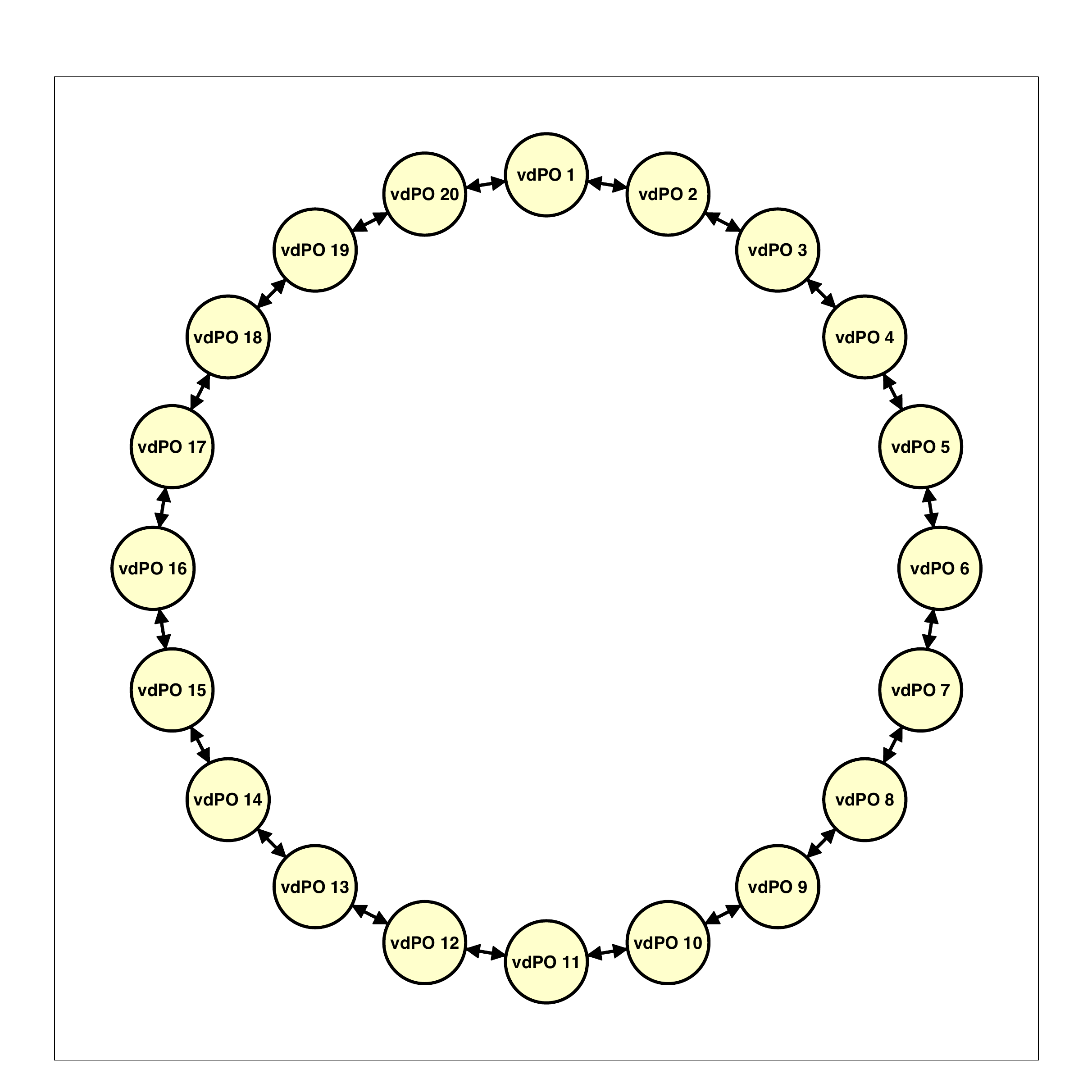}
                 \caption{Ring composed of coupled van der Pol oscillators.}
                 \label{fig:vdPOsRing}
               \end{figure}

               The dynamical model of the $i$-th coupled vdPO ($\subss{\Sigma}{i}^C$) is given by
               \begin{equation}
                 \label{eq:vdPO}
                 \begin{aligned}
                   \subss{\dot{x}}{i,1} &= \subss x {i,2}\\
                   \subss{\dot{x}}{i,2} &= -(1+2\bar\beta)\subss x {i,1} + \bar\beta \subss x {i-1,1} + \bar\beta \subss x {i+1,1} -\bar\alpha(\subss x {i,1}^2-1)\subss x {i,2} + g_i^A(\subss x {i,1})\subss u {i},
                 \end{aligned}
               \end{equation}
               where $g_i^A(\subss x {i,1})=\frac{1}{0.4+0.1\subss x {i,1}^2}$ is the function describing the nonlinear dynamics of an actuator. Each oscillator $i\in\MM$, is a subsystem with state $\subss x i=(\subss x {i,1},\subss x {i,2})$ and input $\subss u i$, where $\subss x {i,1}$ is the displacements of oscillator $i$ with respect to a given equilibrium position on the ring, $\subss x {i,2}$ is the velocity of the oscillator $i$ and $\subss u {i}$ is the force applied to oscillator $i$. For all vdPOs, we consider $\bar\alpha=0.1$ and $\bar\beta=-0.3$. Subsystems are equipped with the state constraints $\norme{\subss x {i,1}}{\infty}\leq 3$, $\norme{\subss x {i,2}}{\infty}\leq 2$, $i\in\MM$ and with the input constraints $\norme{\subss u {i}}{\infty}\leq 8$. We obtain models $\subss\Sigma i$ by discretizing continuous-time models with $T_s=0.1~$sec sampling time, using Euler discretization. In this example, the local fault detectors do not share variables, hence $\subss{\Sigma}{i}=\subss{\tilde\Sigma}{i}$. Moreover the design parameter of fault detectors has been set $\lambda=0.1$. As regards the control architecture, for each controller, we set
               \begin{equation*}
                 \subss u i =  (0.4+0.1\subss x {i,1}^2)\left[\bar\alpha(\subss x {i,1}^2-1)\subss x {i,2} + \subss v i + \bkappa_i(\subss x i-\subss\bx i)\right].
               \end{equation*}
               Then, we synthesize controllers $\subss\CC i$, $i\in\MM$ using Algorithm \ref{alg:pnpcontrollers}.\\
               In the following simulation, we consider a ring composed of $M=20$ vdPOs (see Figure \ref{fig:vdPOsRing}). We also consider the measurement errors bounded in the sets
               \begin{equation*}
                 \Oset_i=\{ \subss \varrho i\in\Rset^{2}:~\norme{\subss \varrho {i}}{\infty}\leq 10^{-1}\}.
               \end{equation*}
               The modelling of the LSS, the design of PnPMPC controllers and the simulations have been performed using the PnPMPC toolbox for MatLab \cite{Riverso2012g}. During the simulation, the control action $\subss u i(t)$ computed by the controller $\CC_{[i]}$, for all $i\in\MM$, is kept constant during the sampling interval and applied to the continuous-time system. In Figure \ref{fig:position} and \ref{fig:velocity} we show a simulation where at $t=0$, each vdPO is placed in a random position around the origin. For $0\leq t< 2.5s$,  due to the presence of measurements errors, the state is kept around the origin. In particular each controller $\subss\CC i$ computes the control inputs shown in Figure \ref{fig:input}. At time $\bar t=2.5s$, a fault occurs in the $11$-th vdPO: the actuator is breakdown and saturates the control input, hence $\subss u {11}(\bar t)=8$, $\forall t\geq \bar t$, and we can also see that the velocity of the $11$-th vdPO diverges from the origin. The next time instant, due to a large error between the state estimates and the measured states, the $11$-th FD detects the fault, indeed $\abs{\subss y {11,2}(\bar t+T_s)-\subss{\hat\tx}{11,2}(\bar t+T_s)}\geq\subss{\bar\epsilon}{11,2}(\bar t+T_s)$ (see Figure \ref{fig:thre}). At this time instant, the reconfiguration process starts: the faulty subsystem is unplugged and then the neighbouring oscillators ($\subss\Sigma{10}$ and $\subss\Sigma{12}$) retune their controllers and their fault detectors. In Figure \ref{fig:position} and \ref{fig:velocity}, we can note that for $t\geq \bar t+10T_s$, the states are still kept around the origin. At time $t=\bar t+10T_s$, the $11$-th actuator is fixed, then the vdPO can be plugged in: therefore neighbouring oscillators retune their controllers and fault detectors. The oscillator is initialized with $\subss x {11}(\bar t + 10T_s)=(2.5,0)$ and then the controller steers the state around the origin.

               \begin{figure}[!htb]
                 \centering
                 \begin{subfigure}[!htb]{0.28\textwidth}                   
                   \includegraphics[scale=0.25]{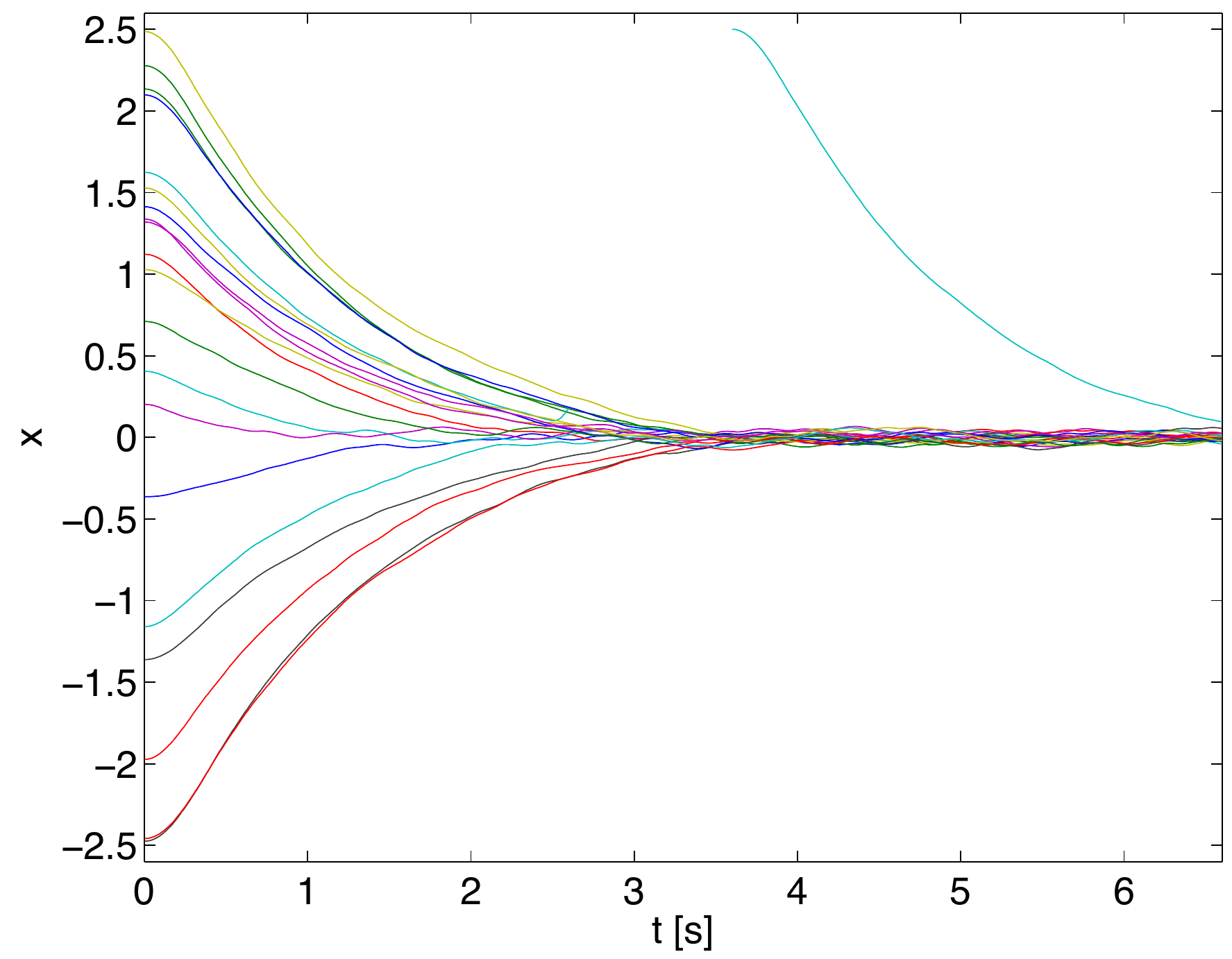}
                   \caption{Displacements of the vdPOs, i.e. states $\subss{x}{i,1}$, $i\in\MM$.}
                   \label{fig:position}
                 \end{subfigure}\quad
                 \begin{subfigure}[!htb]{0.28\textwidth}                   
                   \includegraphics[scale=0.25]{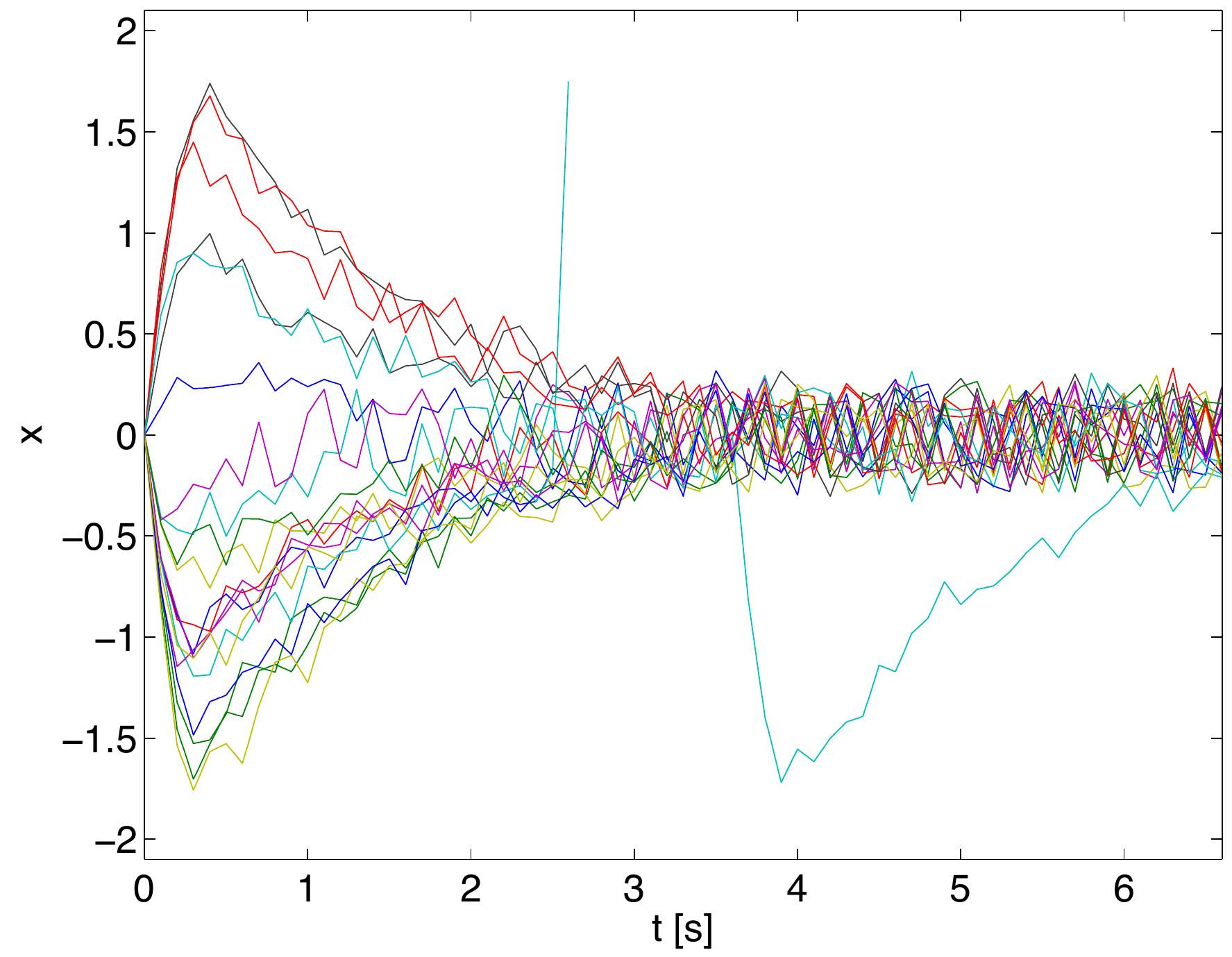}
                   \caption{Velocities, i.e. states $\subss{x}{i,2}$, $i\in\MM$.}
                   \label{fig:velocity}
                 \end{subfigure}\quad
                 \begin{subfigure}[!htb]{0.28\textwidth}                   
                   \includegraphics[scale=0.25]{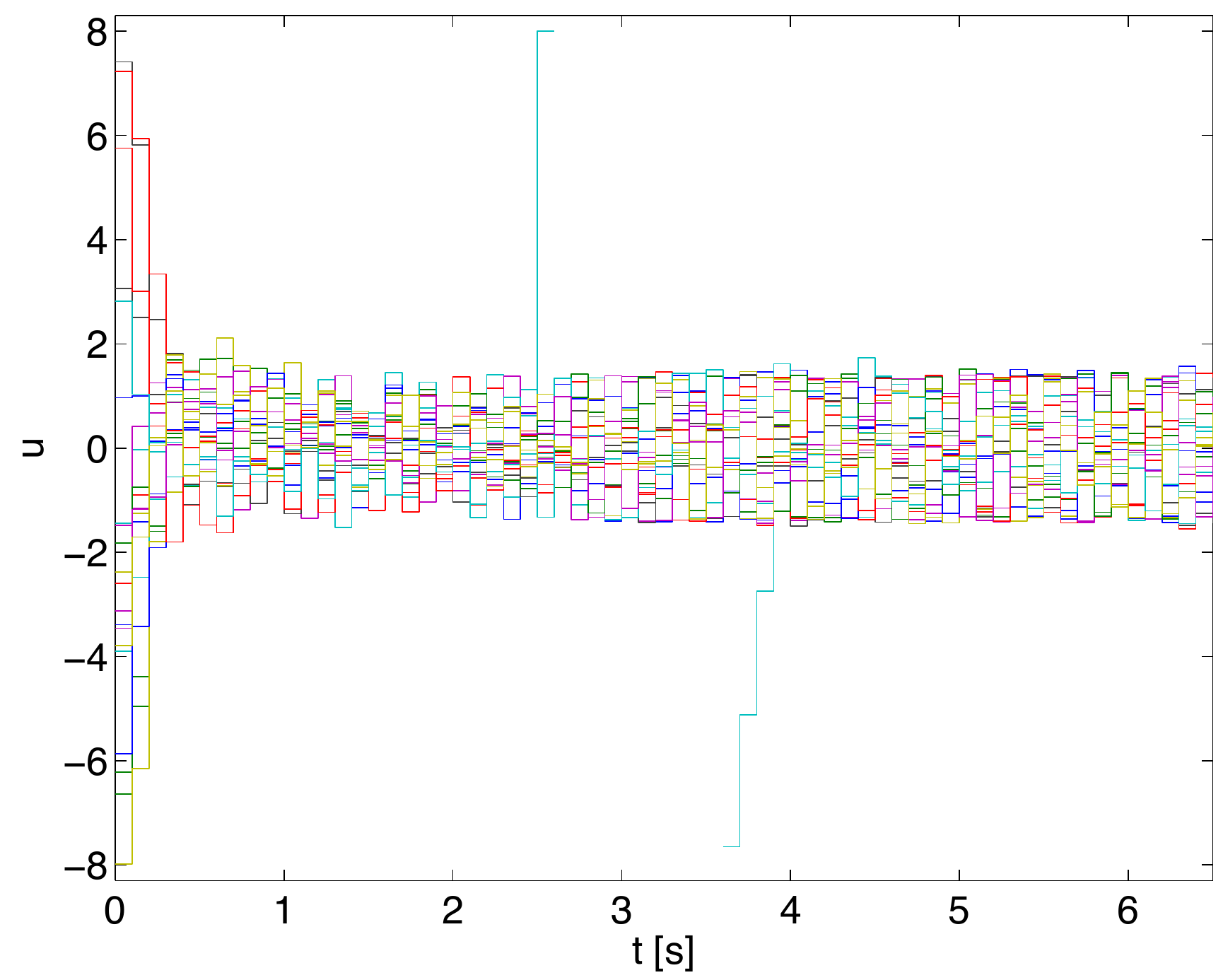}
                   \caption{Inputs $\subss{u}{i}$, $i\in\MM$.\newline}
                   \label{fig:input}
                 \end{subfigure}
               \end{figure}

               \begin{figure}[!htb]
                 \centering
                 \includegraphics[scale=0.42]{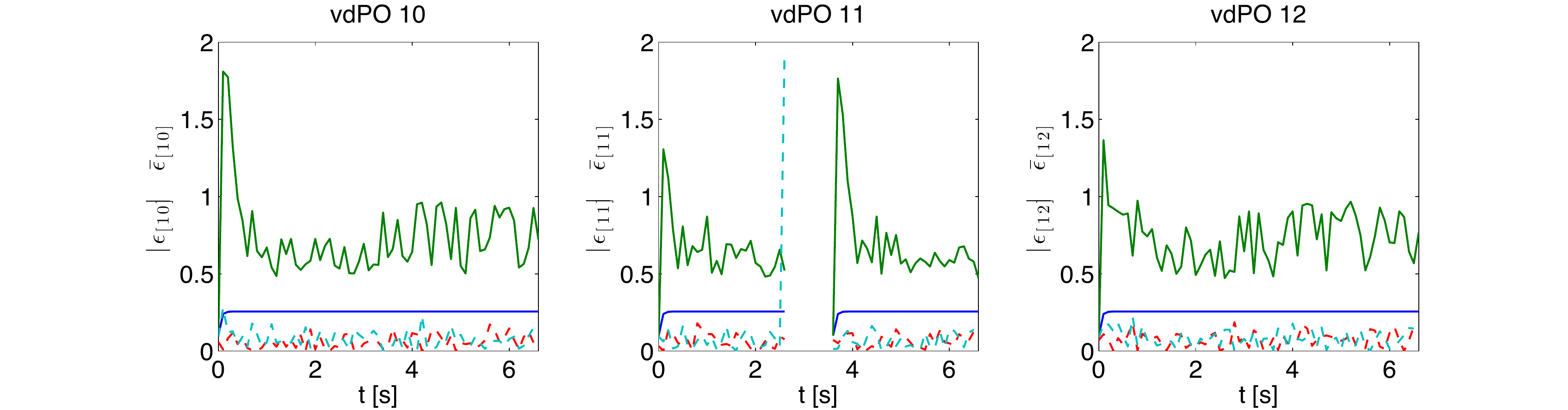}
                 \caption{Dashed lines are the absolute values of errors $\subss\epsilon i=\abs{\subss y i-\subss{\hat{\tilde{x}}}i}$ and bold lines are the thresholds $\subss{\bar\epsilon}i$, for $i=\{10,11,12\}$.}
                 \label{fig:thre}
               \end{figure}
               \begin{figure}[!htb]
                 \centering
                 \includegraphics[scale=0.35]{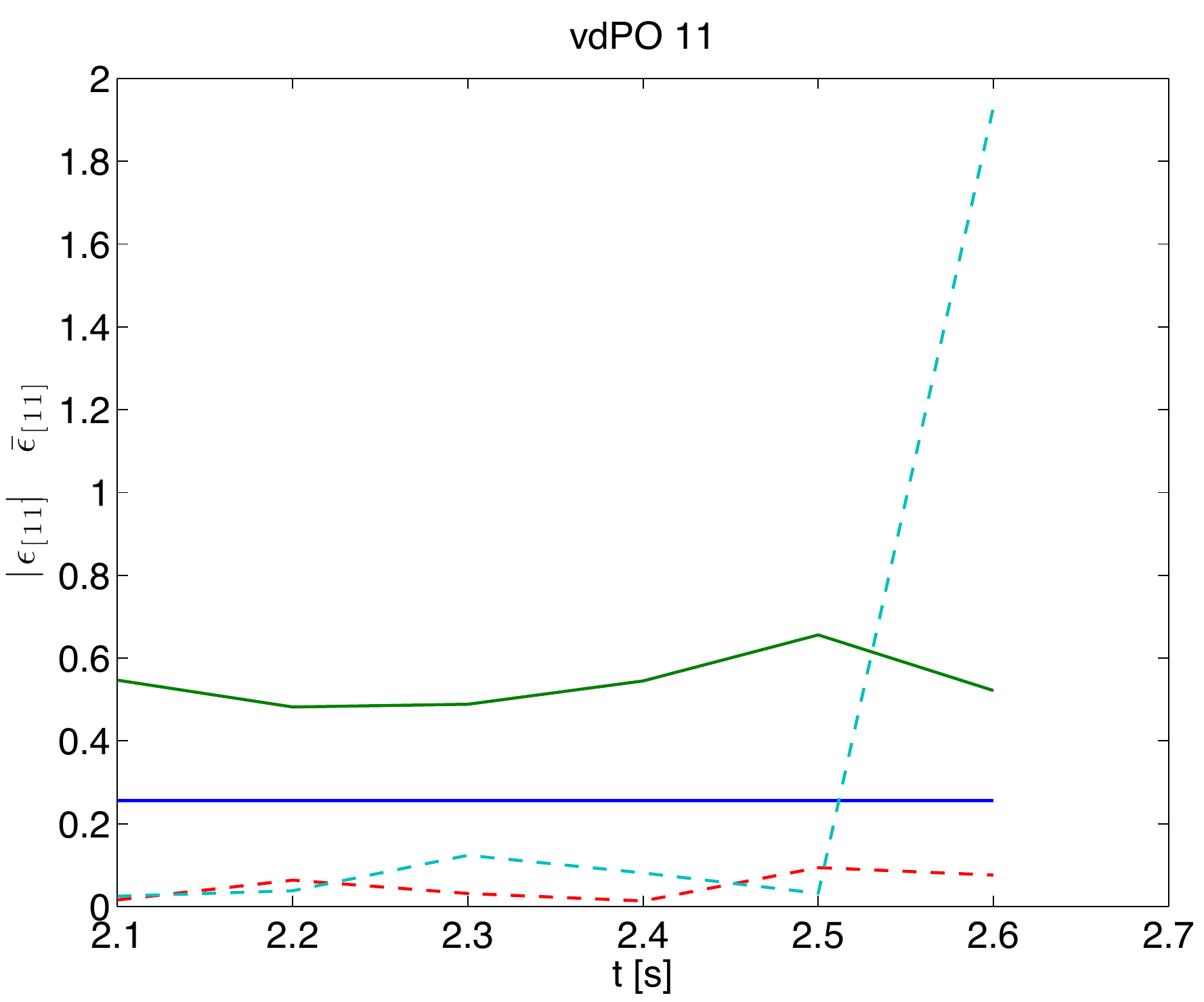}
                 \caption{Dashed lines are the absolute values of errors $\subss\epsilon {11}=\abs{\subss y {11}-\subss{\hat{\tilde{x}}}{11}}$ and bold lines are the thresholds $\subss{\bar\epsilon} {11}$ during the detection of the fault in the 11-th vdPO.}
                 \label{fig:thre}
               \end{figure}

          \subsection{Power Networks System} 
               \label{sec:PNS}
               In this example, we apply the proposed state-feedback PnPMPC and FD scheme to the PNS proposed in Appendix B of \cite{Riverso2014}. In the following we first design the AGC layer for the PNS composed of $5$ areas as in Figure \ref{fig:pns:scenario2}, then we show how, after a fault in area 4, we can disconnect the faulty area (unplugging operation) and redesign the controllers of neighbouring areas (reconfiguration operation).
               \begin{figure}[!htb]
                 \centering
                 \includegraphics[scale=0.65]{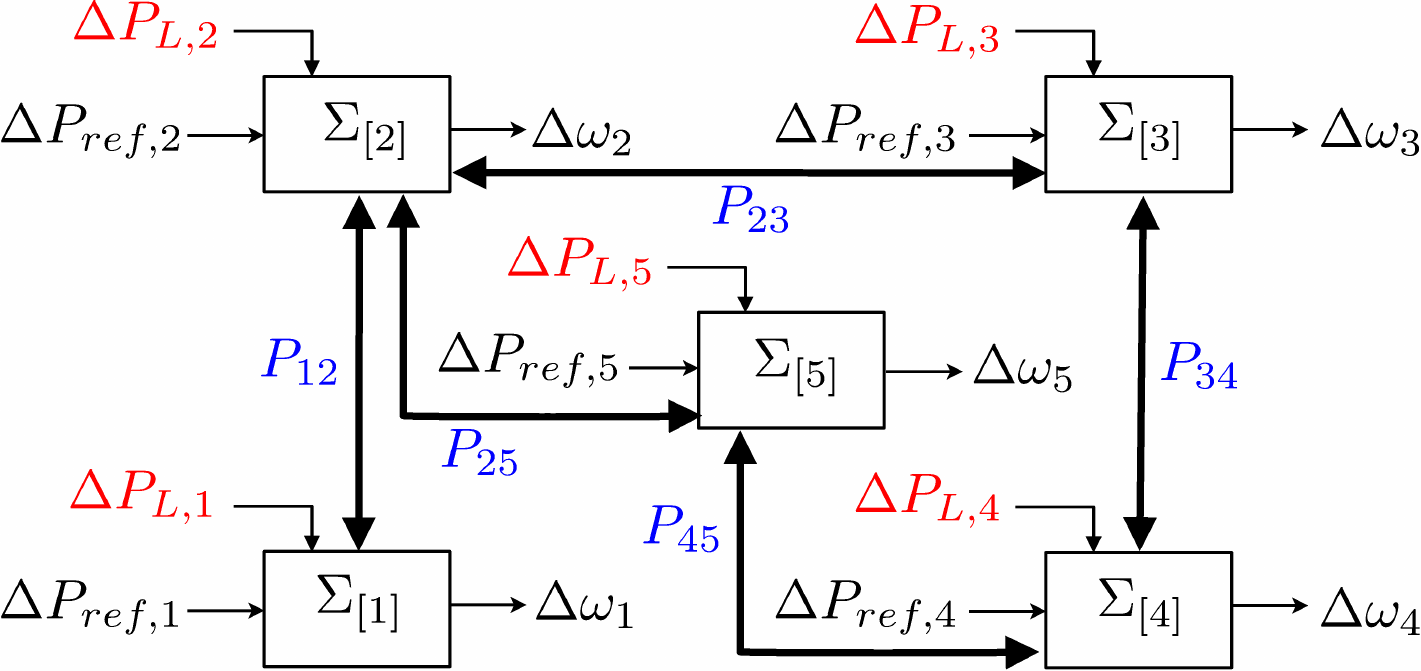}
                 \caption{Power network system of Scenario 2.}
                 \label{fig:pns:scenario2}
               \end{figure}
               
               The dynamics of an area equipped with primary control and linearized around the equilibrium value for all variables can be described by the following model \cite{Saadat2002}
               \begin{equation}
                 \label{eq:ltipower}
                 \subss{\Sigma}{i}^C:\quad\subss{\dot{x}}{i} = A_{ii}\subss x i + B_{i}\subss u i + L_{i}\Delta P_{L_i} + \sum_{j\in\NN_i}A_{ij}\subss x j
               \end{equation}
               where $\subss x i=(\Delta\theta_i,~\Delta\omega_i,~\Delta P_{m_i},~\Delta P_{v_i})$ is the state, $\subss u i = \Delta P_{ref_i}$ is the control input of each area, $\Delta P_L$ is the local power load and $\NN_i$ is the sets of neighbouring areas, i.e. areas directly connected to $\subss\Sigma i^C$ through tie-lines. The matrices of system \eqref{eq:ltipower} are
               $$\begin{array}{c}A_{ii}(\{P_{ij}\}_{j\in\NN_i}) = \matr{ 0 & 1 & 0 & 0 \\ -\frac{\sum_{j\in\NN_i}{P_{ij}} }{2H_i} & -\frac{D_i}{2H_i} & \frac{1}{2H_i} & 0 \\ 0 & 0 & -\frac{1}{T_{t_i}}  & \frac{1}{T_{t_i}} \\ 0 & -\frac{1}{R_iT_{g_i}} & 0 & -\frac{1}{T_{g_i}} }\\
                 B_{i} = \matr{ 0 \\ 0 \\ 0 \\ \frac{1}{T_{g_i}} },\,
                 A_{ij} = \matr{ 0 & 0 & 0 & 0 \\ \frac{P_{ij}}{2H_i} & 0 & 0 & 0 \\ 0 & 0 & 0  & 0 \\ 0 & 0 & 0 & 0 },\,
                 L_{i} = \matr{ 0 \\ -\frac{1}{2H_i} \\ 0 \\ 0 }\end{array}$$
               For the meaning of constants as well as parameter values we refer the reader to Appendix B of \cite{Riverso2014}. We highlight that all parameter values are within the range of those used in Chapter 12 of \cite{Saadat2002}. Model~\eqref{eq:ltipower} is input decoupled since both $\Delta P_{ref_i}$ and $\Delta P_{L_i}$ act only on subsystem $\subss{\Sigma}{i}^C$. Moreover, subsystems $\subss\Sigma i^C$ are parameter dependent since the local dynamics depends on the quantities $-\frac{\sum_{j\in\NN_i}{P_{{ij}}} }{2H_i}$. Each subsystem $\subss{\Sigma}{i}^C$ is subject to constraints on $\Delta\theta_i$ and on $\Delta P_{ref_i}$ specified in Appendix B of \cite{Riverso2014}. We obtain models $\subss\Sigma i$ by discretizing models $\subss\Sigma i^C$ with $1$ sec sampling time, using exact discretization and treating $\subss u i$, $\Delta P_{L_i}$, $\subss x j,~j\in\NN_i$ as exogenous signals. As regards the FDA, each area is equipped with a local FD $\subss{\tilde{\Sigma}}i$ that share some state variables. In particular area 1 and 2 share $\Delta\theta_1$, area 2 and 3 share $\Delta\theta_3$, area 2 and 5 share $\Delta\theta_5$ and area 3, 4 and 5 share $\Delta\theta_4$. We note that the choice of shared variables allow each FD to locally consider the effect of coupling terms and hence, from an electrical point of view, to take into account how tie-line powers are exchanged among areas. Moreover we consider the following bounded measurement errors
               \begin{equation*}
                 \Oset_i=\{ \subss \varrho i\in\Rset^{4}:~\norme{\subss \varrho {i}}{\infty}\leq 10^{-3}\}.
               \end{equation*}
               The modelling of the LSS, the design of PnPMPC controllers and the simulations have been performed using the PnPMPC toolbox for MatLab \cite{Riverso2012g}. For each subsystem $\subss\Sigma i$, the controller $\subss\CC i$, $i\in\MM$ is designed by executing Algorithm \ref{alg:pnpcontrollers}. The aim of the AGC layer is to restore the frequency in each area next to step loads, therefore each controller must be designed in order to stabilize the local area around an equilibria that depends on $\Delta P_{L_i}$. As regards FDA, for each local FD $\subss{\tilde{\Sigma}}i$, the filter parameter $\lambda$ is set to $0.5$.
               
               In the control experiment, step power loads $\Delta P_{L_i}$ specified in Table \ref{tab:pns:simulationscen2} have been used and they cause the step-like changes of the control variables in Figure \ref{fig:simulationscen2}. 
               \begin{table}[!htb]
                 \centering
                 \begin{tabular}{|c|c|c|}
                   \hline
                   Step time  &  Area $i$ & $\Delta P_{L_i}$ \\
                   \hline
                   5               &      1        &   +0.10             \\
                   \hline
                   15             &      2        &   -0.16             \\
                   \hline
                   20             &      1        &   -0.22             \\
                   \hline
                   20             &      2        &   +0.12             \\
                   \hline
                   20             &      3        &   -0.10             \\
                   \hline
                   30             &      3        &   +0.10             \\
                   \hline
                   40             &      4        &   +0.08             \\
                   \hline
                   40             &      5        &   -0.10             \\
                   \hline
                 \end{tabular}
                 \caption{Load of power $\Delta P_{L_i}$ (p.u.) for simulation. $+\Delta P_{L_i}$ means a step of required power, hence a decrease of the frequency deviation $\Delta\omega_i$ and therefore an increase of the power reference $\Delta P_{ref_i}$.}
                 \label{tab:pns:simulationscen2}
               \end{table}
               In Figure \ref{fig:simulationscen2freq} we show, how in presence of loads, the frequency deviation is steered in a neighbourhood of zero: however, due to the presence of measurement errors $\subss\varrho i$ (randomly extracted in the sets $\Oset_i$), $\Delta\omega_i$ cannot be perfectly zeroed. In Figure \ref{fig:simulationscen2ref} we note how the power references $\Delta P_{ref_i}$ are changed in order to compensate for local loads. 
               
               \begin{figure}[!htb]
                 \begin{subfigure}[!htb]{1\textwidth}
                   \centering
                   \includegraphics[scale=0.44]{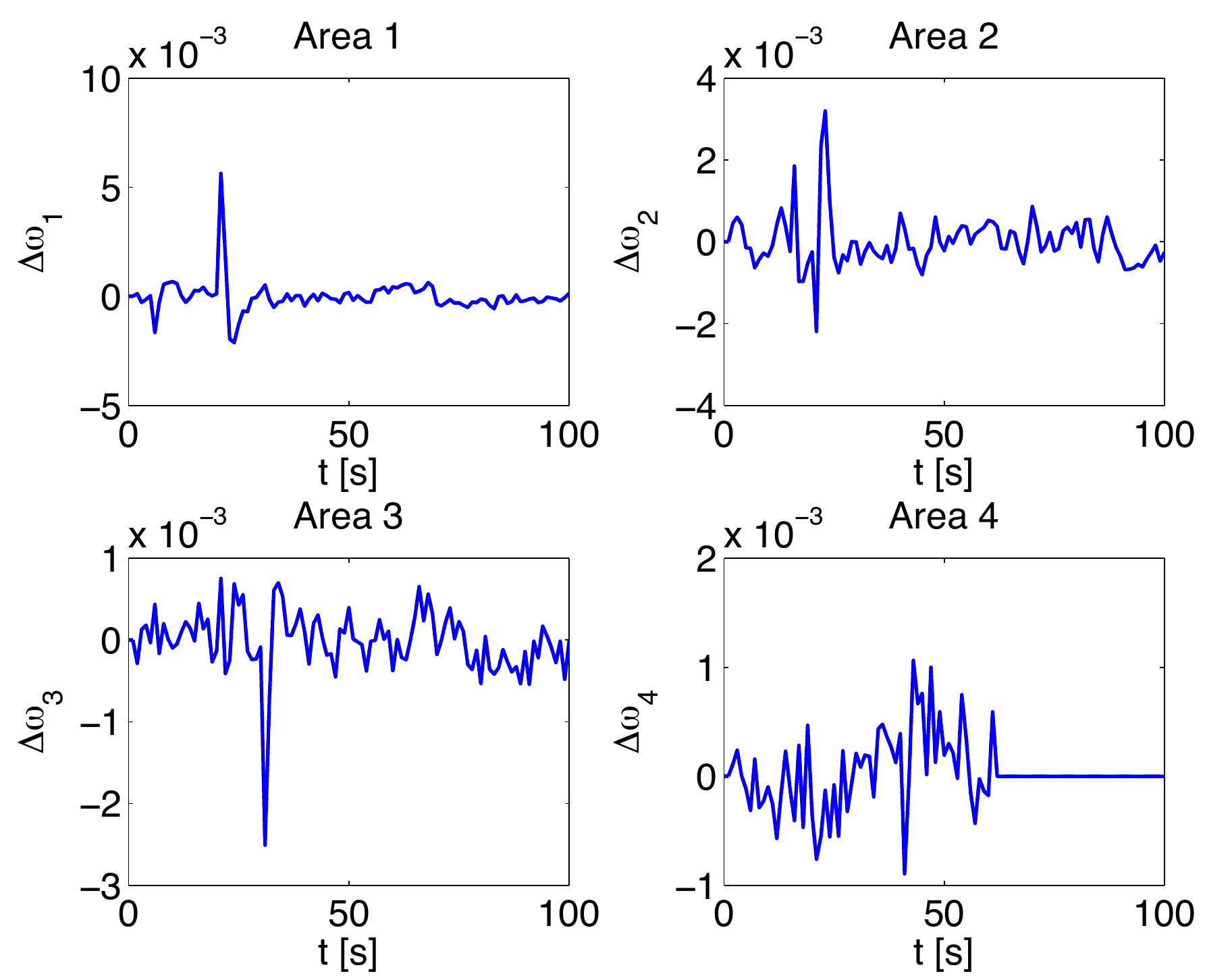}   
                   \includegraphics[scale=0.42]{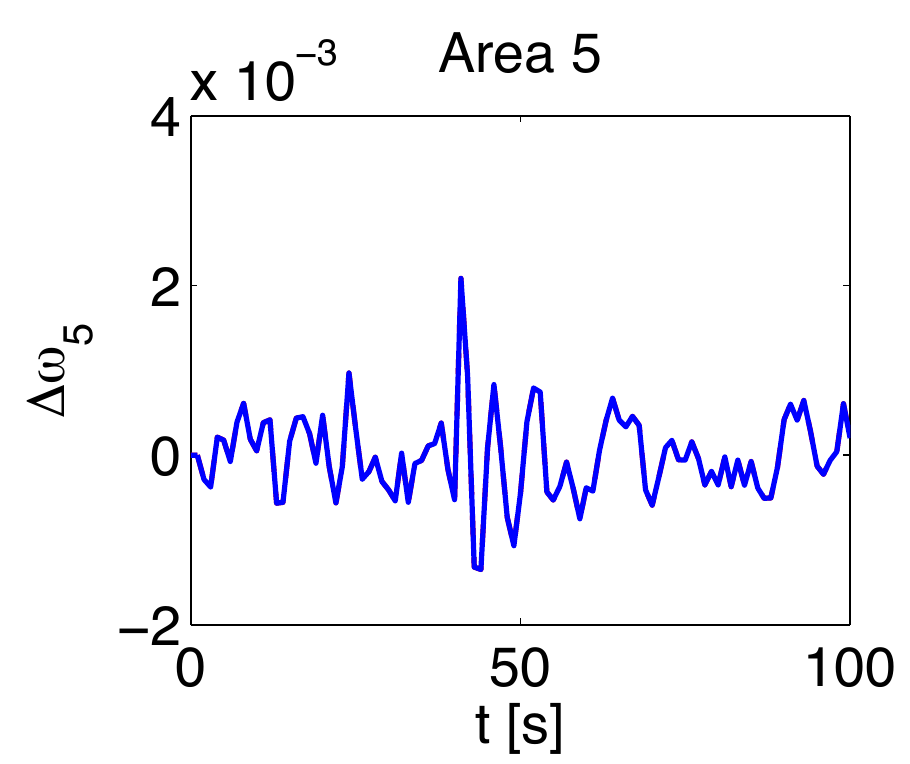}
                   \caption{Frequency deviation in each area controlled by PnPMPC controllers. Note that $\Delta\omega_4=0$ after unplugging of area 4.}
                   \label{fig:simulationscen2freq}
                 \end{subfigure}
                 \begin{subfigure}[!htb]{1\textwidth}
                   \centering
                   \includegraphics[scale=0.44]{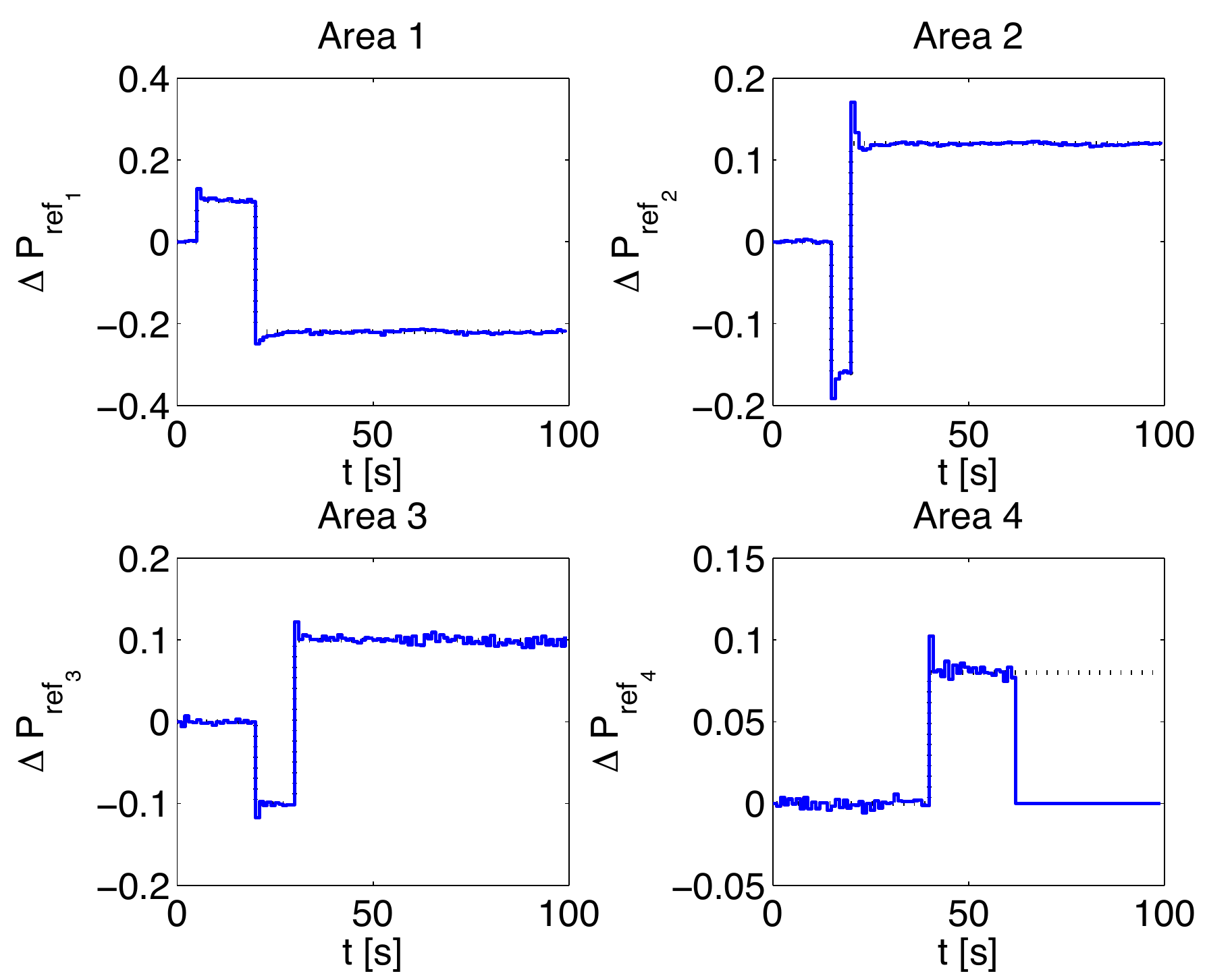}
                   \includegraphics[scale=0.42]{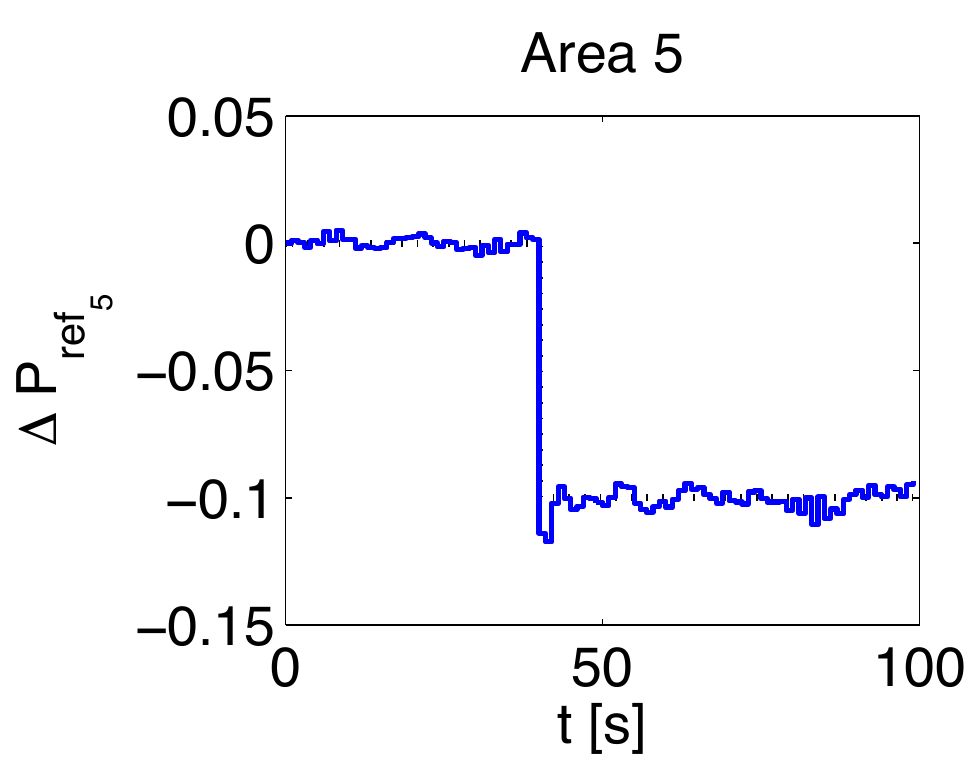}
                   \caption{Load reference set-point in each area controlled by PnPMPC controllers. Note that $\Delta P_{ref_4}=0$ after unplugging of area 4.}
                   \label{fig:simulationscen2ref}
                 \end{subfigure}            
                 \caption{Simulation of a fault in area $4$ at time $t=62$: frequency deviation (panel \ref{fig:simulationscen2freq}) and load reference (panel \ref{fig:simulationscen2ref}) in each area.}
                 \label{fig:simulationscen2}
               \end{figure}
               
               \begin{figure}[!htb]
                 \centering
                 \includegraphics[scale=0.54]{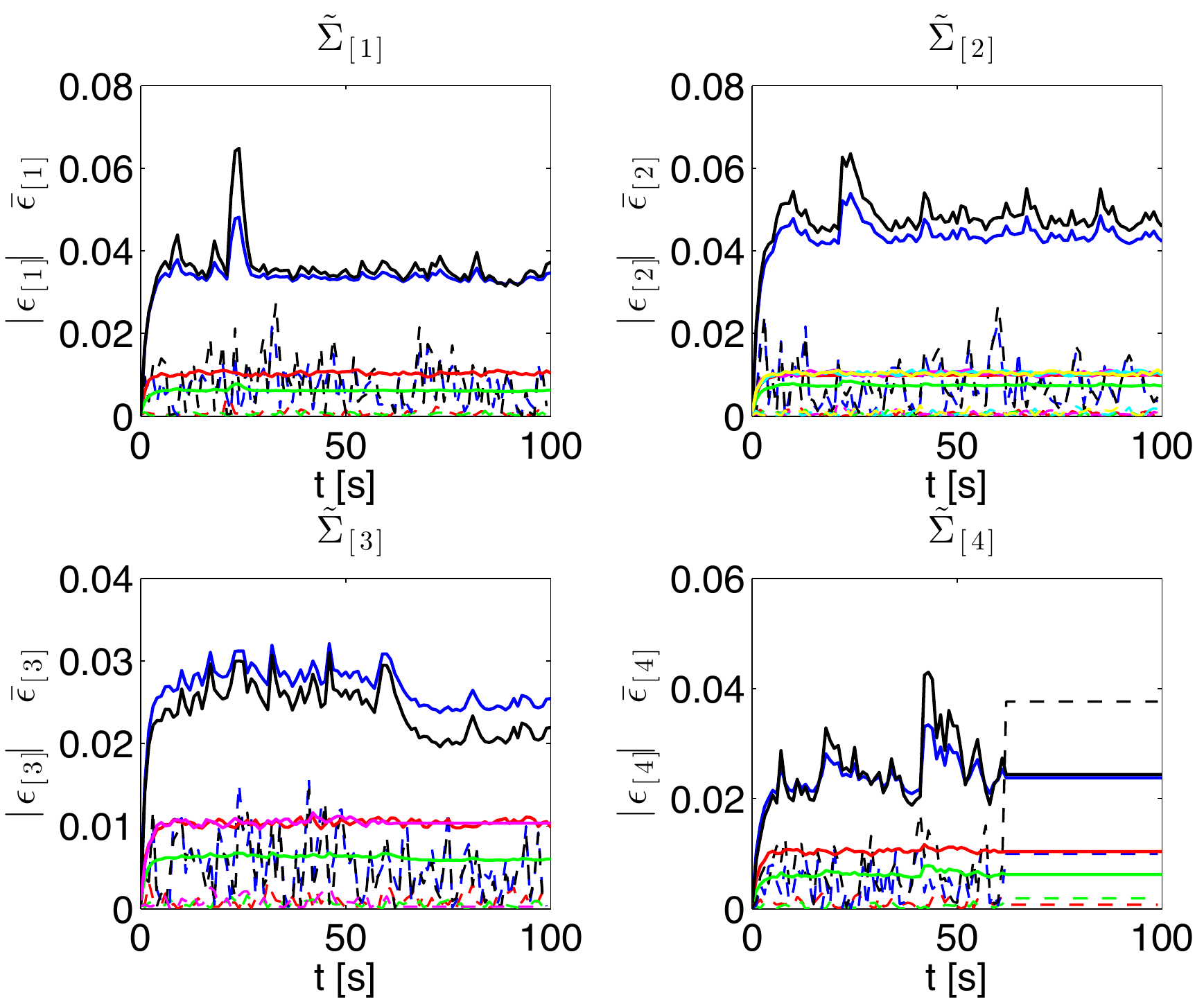}
                 \includegraphics[scale=0.52]{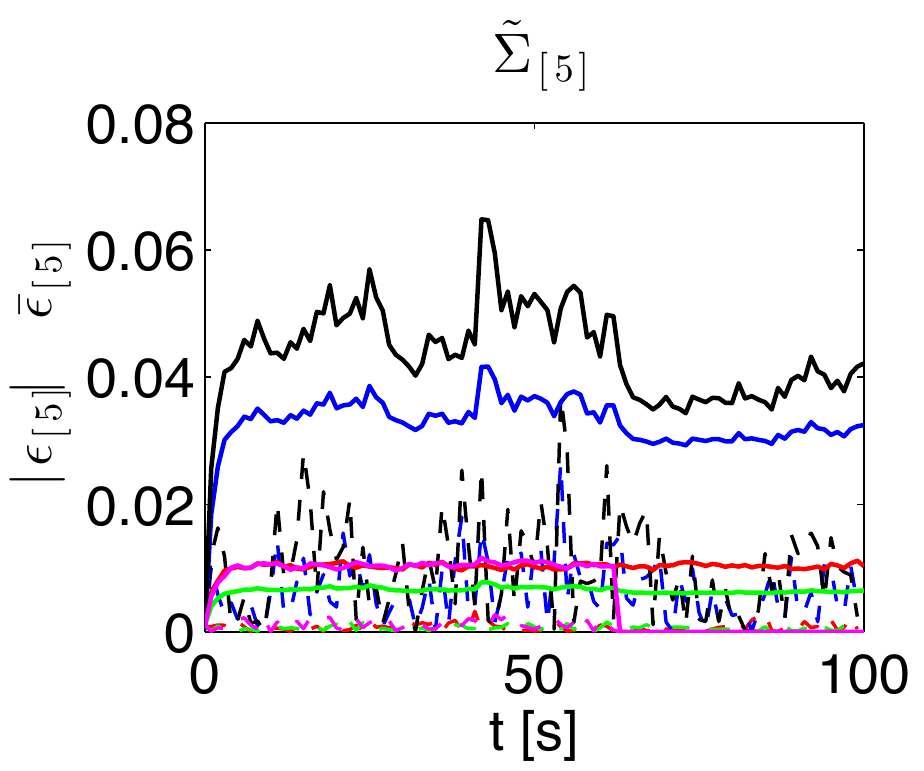}
                 \caption{Simulation: for each area, dashed lines are the absolute values of errors $\subss\epsilon i=\subss y i-\subss{\hat{\tilde{x}}}i$ and bold lines are the thresholds $\subss{\bar\epsilon}i$.}
                 \label{fig:simulationscen2thresh}
               \end{figure}
                             
               At time instant $t=60$, the following fault occurs in area 4: the inertia constant $H_4$ is reduced from $8$ to $1$. From an electrical point of view, there is a fault in a local generator, hence, for safety reasons, area 4 must be isolated in order to not propagate faults in the PNS. In Figure \ref{fig:simulationscen2thresh}\footnotemark[3], we note that for $t<62$, the errors $\abs{\subss\epsilon i}$ are always upper bounded by the thresholds $\subss{\bar\epsilon} i$, hence no faults are detected. At time instant $t=62$, FD $\subss{\tilde\Sigma}4$ detects the fault in area 4, indeed $\abs{\subss\epsilon{\Delta P_{v_4}}}(62)>\subss{\bar\epsilon}{\Delta P_{v_4}}(62)$. Therefore, area 4 is unplugged and controllers $\subss\CC i$ and FDs $\subss{\tilde\Sigma}i$, $i=\{3,5\}$ are retuned. Note that the reconfiguration operation does not involve areas 1 and 2 since they where not connected with area 4 and they did not share any state variables with it. Therefore the reconfiguration process is not propagated in the network. Next to the unplugging of area 4, the new PNS can still compensate power loads and FDs do not detect any fault\footnotemark[3].

               \footnotetext[3]{For the convenience of the reader, in Figure \ref{fig:simulationscen2thresh}, after the reconfiguration process, errors and thresholds involving state variables of area 4 are kept constants for display purposes. After fault detection, the local estimator is stopped.}

     \section{Concluding Remarks}
          \label{sec:conclusion}
	
          In this paper, a novel integrated architecture composed of a decentralized MPC scheme and of a distributed FD architecture has been proposed in the context of fault-tolerant control for a class of large-scale nonlinear systems. The integrated control scheme guarantees closed-loop asymptotic stability and constraints satisfaction at each time instant, while the FD architecture allows to detect faulty subsystems guaranteeing the absence of false-alarms and the convergence the estimators also during reconfiguration processes. The innovative idea is to combine distributed MPC and distributed FD architectures, where local controllers and state estimators can be designed in a PnP fashion, i.e. the overall model of the LSS is never used in any step of the design phase. The proposed architecture is suitable for several large-scale applications, allowing revamping of actuators and isolating faulty subsystems before the fault is propagate in the network. Future research efforts will be devoted to generalizing the approach to a larger class of nonlinear systems and to apply it to different types of LSSs.

     \appendix
     \section{Proof of Theorem \ref{thm:ch9:mainclosedloop}}
          \label{sec:prooftheoremctrl}

          \begin{proof}
            The proof of Theorem \ref{thm:ch9:mainclosedloop} is an adaptation of the proof of Theorem 9 in \cite{Riverso2014a} to the non-linear case. Due to space limitation in \cite{Riverso2014a}, this proof is available in \cite{Riverso2014} as the proof of Theorem 6.1. First, we can easily prove that, if $\subss x i(0)\in\Xset_i^N$, the MPC-$i$ optimization problem defined n \eqref{eq:decMPCProblem} is always feasible and its optimizers $\subss\hx i(0|t)$ and $\subss v i(0|t)$ verify $\subss\hx i(0|t)\rightarrow \Zero_{n_i}$ and $\subss v i(0|t)\rightarrow \Zero_{m_i}$ as $t\rightarrow\infty$. 
            
            Differently from \cite{Riverso2014}, where coupling terms have been defined as linear functions, subsystems $\subss\Sigma i$, $i\in\MM$ defined in this paper take into account nonlinearities in the coupling among subsystems.

            Similarly to Step 1 of the proof of Theorem 6.1 in \cite{Riverso2014}, we aim at showing that if $\subss x i(0)\in\Xset_i^N$ there is $\tilde T>0$ such that $\subss x i(\tilde T)\in\Zset_i$ and hence $\dist{\Zset_i}{\subss x i(\tilde T)}=0$. From \eqref{eq:subsystem} and \eqref{eq:NLtubecontrol}, we can write
            \begin{equation}
              \label{eq:ch7:controlled_model_rewrite}
              \subss x i(t+1)=A_{ii}\subss x i(t)+B_i\bkappa_i(\subss x i(t))+w_i(\subss\psi i(t))+{\bar\eta_i}(t)
           \end{equation}
           where            
           \begin{equation}
             \label{eq:ch7:bareta}
             {\bar\eta_i}(t)={B_i}( \subss v i(t)+\bkappa_i(\subss z i(t))-{\bkappa_i}(\subss x i(t)) )
           \end{equation}
           and $\subss z i(t)=\subss x i(t)-\subss\hx i(0|t)$. In particular, if $\subss x i(0)\in\Xset_i^N$, recursive feasibility of the MPC-$i$ problem \eqref{eq:decMPCProblem} implies that \eqref{eq:ch7:controlled_model_rewrite} holds for all $t\geq 0$.\\
           Note that Step (\ref{enu:AssOmegai}) of Algorithm \ref{alg:pnpcontrollers} guarantees that Assumption 6.3 in \cite{Riverso2014} is verified and therefore, the LP problem (6.14) in \cite{Riverso2014} is feasible for all $\subss z i\in\Rset^{n_i}$. This implies that the function $\bkappa_i(\subss x i(t))$ in \eqref{eq:ch7:controlled_model_rewrite} is always well defined.\\
           From the asymptotic convergence to zero of the nominal state $\subss \hx i(0|t)$ and the input signal $\subss v i(0|t)$, it holds
           \begin{equation}
             \label{eq:ch7:hxvdelta}
             \forall\delta_i>0,~\exists T_{i,1}>0:~\norme{\subss \hx i(0|t)}{}\leq\delta_i\mbox{ and }\norme{\subss v i(0|t)}{}\leq\delta_i,~\forall t\geq T_{i,1}.
           \end{equation}
           Moreover, according to \cite{Gal1995}, we can assume without loss of generality that $\bkappa_i(\cdot)$ is a continuous piecewise affine map. In view of this, $\bkappa_i(\cdot)$ is also globally Lipschitz, i.e.
           \begin{equation}
             \label{eq:ch7:lipschitz}
             \exists~L_i>0~:~\norme{\bkappa_i(\subss x i-\subss\hx i)-\bkappa_i(\subss x i)}{}\leq L_i\norme{\subss\hx i}{}
           \end{equation}
           for all $(\subss x i,\subss\hx i)$ such that $\subss x i\in\Xset_i$ and $\subss x i-\subss\hx i\in\Zset_i$. Using \eqref{eq:ch7:lipschitz} one can show that setting $\delta_i=\frac{\epsilon_i}{\norme{B_i}{}(1+L_i)}$ the following implication holds for all $\epsilon_i>0$:
           \begin{equation*}
             \label{eq:ch7:etaepsilon}
             \norme{\subss \hx i(0|t)}{}\leq\delta_i\mbox{ and }\norme{\subss v i(0|t)}{}\leq\delta_i\Rightarrow\norme{{\bar\eta_i}(t)}{}\leq\epsilon_i,~\forall \subss x i(t)\in\Xset_i.
           \end{equation*}
           Therefore, from \eqref{eq:ch7:hxvdelta},
           \begin{equation}
             \label{eq:ch7:epsilonT}
             \forall\epsilon_i>0,~\exists T_{i,1}>0:~\norme{{\bar\eta_i}(t)}{}\leq\epsilon_i,~\forall t\geq T_{i,1}.
           \end{equation}
           Since $\norme{\subss\hx i(0|t)}{}\rightarrow \Zero_{n_i}$, as $t\rightarrow\infty$, and $\Zset_i$ contains $\ball{\omega_i}(\Zero_{n_i})$ (see Step (\ref{enu:AssOmegai}) of Algorithm \ref{alg:pnpcontrollers}), then
           \begin{equation}
             \label{eq:ch7:deltaz}
             \forall\delta_{z_i}>0,~\exists T_{i,2}>0:~\subss\hx i(0|t)\in\delta_{z_i}\Zset_i,~\forall t\geq T_{i,2}
           \end{equation}
           Hence, from \eqref{eq:inZproblem},
           \begin{equation}
             \label{eq:ch7:xafterT2}
             \subss x i(t)=\subss \hx i(0|t)+(\subss x i(t)-\subss\hx i(0|t))\in(1+\delta_{z_i})\Zset_i,~\forall t\geq T_{i,2}.
           \end{equation}
           From \eqref{eq:ch7:controlled_model_rewrite} we have, for all $i\in\MM$,
           \begin{equation}
             \label{eq:ch7:fromCollectivemodel}
             \subss x i(t+1)=A_{ii}\subss x i(t)+B_i\bkappa_i(\subss x i(t))+\subss{\hat w}i(t)
           \end{equation}
           where $\subss{\hat w} i=w_i(\subss\psi i)+\subss{\bar\eta} i$, $\forall i\in\MM$. Let $\PP_i$ be the map that builds the vector $\subss\psi i$ from $\{\subss x j\}_{j\in\NN_i}$, i.e. $\subss\psi i=\PP_i(\{\subss x j\}_{j\in\NN_i})$ and define $\subss{\hat\psi}i=\{\PP_i(\{\subss x j\}_{j\in\NN_i}):~\subss x j\in (1+\delta_{z_i})\Zset_j\}$. Setting $\bar T=\max_{i\in\MM}\{T_{i,1},T_{i,2}\}$ and $\delta_z=\max_{i\in\MM}\delta_{z_i}$, using \eqref{eq:ch7:epsilonT} and \eqref{eq:ch7:xafterT2}, remembering that $\subss\psi i$ is the vector of coupling variables, one has, $\forall t\geq\bar T$
           \begin{equation}
             \label{eq:ch7:tildewin1plusdeltaz}
             \subss{\hat w} i\in w_i(\hat\Psi_i)\oplus\ball\epsilon_i(\Zero_{n_i}).
           \end{equation}
           From Steps (\ref{enu:AssOmegai})-(\ref{enu:ch9:hXVsetAlg}) of Algorithm \ref{alg:pnpcontrollers}, since $\Psi_i=\{\PP_i(\{\subss x j\}_{j\in\NN_i}):~\subss x j\in\Xset_j\}$, using \eqref{eq:tightconstraint}, we can deduce that $\hat\Psi_i\subset\dot\Psi_i$. Under Assumption \ref{ass:standardAssumCtrl}-(\ref{ass:couplinglimits}), we have
           \begin{subequations}
             \begin{align}
               \label{eq:ch7:setPredecessorsa}&w_i(\hat\Psi_i)\subset \Wset_i=w_i(\dot\Psi_i)
             \end{align}
           \end{subequations}
           Therefore, there is $\xi_i\in[0,1)$ (that does not depend on $\epsilon_i$) such that
           \begin{equation}
             \label{eq:ch7:sumZjinepsW}
             w_i(\hat\Psi_i)\subseteq\xi_i\Wset_i,
           \end{equation}
           and then, from \eqref{eq:ch7:tildewin1plusdeltaz},
           \begin{equation*}
             \label{eq:ch7:tildewindeltazepsW}
             \subss{\hat w}i\in(1+\delta_z)\xi_i\Wset_i\oplus\ball{\epsilon_i}(\Zero_{n_i}),~\forall t\geq\bar T.
           \end{equation*}
           Note that in \eqref{eq:ch7:epsilonT} the parameter $\epsilon_i>0$ can be chosen arbitrarily small. Assume that it verifies $\epsilon_i<(1+\delta_z)\xi_i\bar\omega_i$, $\forall i\in\MM$ where $\bar\omega_i$ are the radii of the balls in Assumption 6.3 in \cite{Riverso2014}. Then, using Assumption 6.3 in \cite{Riverso2014} we get for $t\geq\bar T$
           \begin{equation}
             \label{eq:ch7:wtildei}
             \subss{\hat w}i(t)\in(1+\delta_z)\xi_i(\Wset_i\oplus\ball{\bar\omega_i}(\Zero_{n_i}))\subseteq(1+\delta_z)\xi_i\bZset_i^0.
           \end{equation}
           In view of \eqref{eq:ch7:xafterT2} and \eqref{eq:ch7:wtildei}, Lemma 6.2 in \cite{Riverso2014} guarantees that
           \begin{equation}
             \label{eq:ch7:plusxindeltazZ}
             \subss\xp i\in (1+\delta_z)(\Zset_i\ominus(1-\xi_i)\bZset_i^0)
           \end{equation}
           From Assumption 6.3 in \cite{Riverso2014}, one has $\Zset_i\ominus (1-\xi_i)\bZset_i^0\subset\Zset_i\ominus\ball{(1-\xi_i)\omega_i}(\Zero_{n_i})$ and hence, since $\Zset_i$ contains the origin in its interior, there is $\mu_i\in [0,1)$ such that $\Zset_i\ominus (1-\xi_i)\Zset_i^0\subset\mu_i\Zset_i$. From \eqref{eq:ch7:plusxindeltazZ} we get $\subss\xp i\in (1+\delta_z)\mu_i\Zset_i$.
           If in \eqref{eq:ch7:deltaz} we set $\delta_z$ such that $(1+\delta_z)\mu_i<1$, we have shown that for $t=\bar T$ it holds
           $
           \subss x i(\bar T+1)\in\Zset_i
           $
           and Step 1 of the proof is concluded setting $\tilde T=\bar T+1$. \\
           We highlight that proof of Theorem \ref{thm:ch9:mainclosedloop} can be concluded using Steps 2 and 3 of the proof of Theorem 6.1 in \cite{Riverso2014}. In particular in Step 2 we prove the convergence of the overall state to the origin and in Step 3 we prove stability of the closed-loop overall system. We note that Steps 2 and 3 use the fact that set $\Zset=\prod_{i\in\MM}\Zset_i$ is an RCI set for the overall closed-loop system.
         \end{proof}

     \bibliographystyle{IEEEtran}
     \bibliography{FD_PnP_MPC-report}

\end{document}